\documentclass[openacc]{rstransa}

\newtheorem{theorem}{\bf Theorem}[section]

\newtheorem{definition}{\bf Definition}[section]
\newtheorem{lemma}{\bf Lemma}[section]
\newtheorem{conjecture}{\bf Conjecture}[section]

\DeclareMathOperator{\cutrk}{cutrk}
\DeclareMathOperator{\rank}{rank}
\DeclareMathOperator{\rwd}{rwd}
\DeclareMathOperator{\cwd}{cwd}
\DeclareMathOperator{\width}{width}
\DeclareMathOperator{\qr}{qr}

\newcommand{\abs}[1]{|#1|}
\newcommand{\ket}[1]{|#1\rangle}
\newcommand{\bra}[1]{\langle#1|}
\newcommand{\adj}[1]{\textproc{adj}(#1)}

\usepackage{algpseudocode}
\usepackage{algorithm} 

\begin{document}

\title{Transforming graph states using single-qubit operations}

\author{
Axel Dahlberg$^{1}$ and Stephanie Wehner$^{1}$}

\address{$^{1}$QuTech, Lorentzweg 1, 2628 CJ Delft, Netherlands}

\subject{Quantum Information Theory}

\keywords{graph states, single-qubit Clifford operations, single-qubit Pauli measurements, computational complexity, rank-width}

\corres{Axel Dahlberg\\
\email{e.a.dahlberg@tudelft.nl}}

\begin{abstract}
Stabilizer states form an important class of states in quantum information, and are of central importance in quantum error correction. 
Here, we provide an algorithm for deciding whether one stabilizer (\emph{target}) state can be obtained 
from another stabilizer (\emph{source}) state by single-qubit Clifford operations ($\mathrm{LC}$), single-qubit Pauli 
measurements ($\mathrm{LPM}$), and classical communication ($\mathrm{CC}$) between sites holding the individual qubits. 
What's more, we provide a recipe to obtain the sequence of $\mathrm{LC}+\mathrm{LPM}+\mathrm{CC}$ operations which prepare
the desired target state from the source state, and show how these operations can be applied in parallel to reach the target state in constant time.
Our algorithm has applications in quantum networks, quantum computing, and can also serve as a design tool - for example, 
to find transformations between quantum error correcting codes. We provide a software implementation of our algorithm that 
makes this tool easier to apply. 

A key insight leading to our algorithm is to show that the problem is equivalent to one in graph theory, which is to decide whether some graph $G'$ is a \emph{vertex-minor} of another graph $G$.
The vertex-minor problem is in general $\mathbb{NP}$-Complete, but can be solved efficiently on graphs which are not too complex.
A measure of the complexity of a graph is the \emph{rank-width} which equals the \emph{Schmidt-rank width} of a subclass of stabilizer states called graph states, and thus intuitively is a measure of entanglement.  
Here we show that the vertex-minor problem can be solved in time $O(|G|^3)$ where $|G|$ is the size of the graph $G$, whenever the rank-width of $G$ and the size of $G'$ are bounded. Our algorithm is based on techniques by Courcelle for solving fixed parameter tractable problems, where here the relevant fixed parameter is the rank width. The second half of this paper serves as an accessible but far from exhausting introduction to these concepts, that could be useful for many other problems in quantum information. 
\end{abstract}


\begin{fmtext}


\end{fmtext}


\maketitle

\section{Introduction}
Stabilizer states form a well studied class of quantum states with many applications in quantum information theory and technology.
One of the important features of stabilizer states is that they can be described efficiently. As a consequence, quantum circuits 
consisting only of Clifford operations, and thus mapping stabilizer states to stabilizer states, can be simulated efficiently on a classical 
computer~\cite{Gottesman1999}. Despite this, the class of stabilizer states is still rich enough to be useful 
in almost any application of quantum networks. An example of such a state that features in many quantum protocols is the GHZ state. 
In networks, applications thus include quantum secret sharing~\cite{Markham2008}, anonymous transfer~\cite{Christandl2005}, conference key agreement~\cite{Ribeiro2017}, and clock synchronization~\cite{Jozsa2000}.
Furthermore, stabilizer states are also important for quantum computation since they form a universal resource for measurement-based quantum computation~\cite{Raussendorf2001}.
The set of stabilizer states includes cluster states and logical states (codewords) of most 
quantum error correcting codes~\cite{Gottesman1997}.

Here, we provide an algorithm for deciding whether a specific multipartite entangled target state can be obtained from a given 
source state, by local Clifford operations and Pauli measurements on the individual qubits, and classical communication between the sites holding these qubits. 
More specifically, we consider the problem of deciding whether some stabilizer (\emph{target}) state $\ket{S_t}$ can be produced from some given stabilizer (\emph{source}) state $\ket{S_s}$ by single-qubit Clifford operations, single-qubit Pauli measurement and classical communication: $\mathrm{LC}+\mathrm{LPM}+\mathrm{CC}$.
An efficient algorithm for solving this task has many interesting applications, including:
\begin{itemize}
\item Given a source state in a quantum network, or distributed quantum processor, our algorithm allows us to quickly decide if, and how, simple $\mathrm{LC}+\mathrm{LPM}+\mathrm{CC}$ operations can turn this existing source state into a desired target state required by a quantum protocol. This provides a useful tool for entanglement routing and management in a quantum internet, or a networked quantum computer, where fast decisions are essential to combat the limited lifetime of quantum memories.  
\item Quantum error correcting codes play an important role in realizing fault tolerant quantum computing systems, and have also been put forward as a means to eventually realize all photonic quantum repeaters in the future~\cite{Azuma2015}.
Here, our algorithm can be used a design tool to find the measurements to realize a specific target state from a given source state. 
\end{itemize}

We will phrase our results in terms of graph states, which enables us to apply existing algorithmic methods more easily.
Graph states form a strict subclass of stabilizer states, which are described by simple undirected unweighted graphs.
The vertices in the graph represent qubits initialized in the state $\ket{+}=\frac{1}{\sqrt{2}}(\ket{0}+\ket{1}$ and edges represent controlled phase gates between the corresponding qubits.
Any stabilizer state is in fact single-qubit Clifford equivalent to some graph state~\cite{VandenNest2004}.
Furthermore, a graph state that is single-qubit Clifford equivalent to a given stabilizer state on $n$ qubits can be found efficiently in time $\mathcal{O}(n^3)$.
Thus, if we find an efficient algorithm that decides if some graph state $\ket{G_t}$ can be reached from $\ket{G_s}$ by $\mathrm{LC}+\mathrm{LPM}+\mathrm{CC}$ we also have an efficient algorithm for the more general case where the target and source states are stabilizer states.

Let us now phrase our problem in a way that forms a natural relation to graph properties. 
To this end, we introduce the notion of a qubit-minor in the following definition, which precisely captures whether a graph state can be reached from another by $\mathrm{LC}+\mathrm{LPM}+\mathrm{CC}$.
\begin{definition}[qubit-minor]\label{def:QM} Assume $\ket{G}$ and $\ket{G'}$ are graph states on the sets of qubits $V$ and $U$ respectively.
    $\ket{G'}$ is called a qubit-minor of $\ket{G}$ if there exists a sequence of single-qubit Clifford operations (LC), single-qubit Pauli measurements (LPM) and classical communication (CC) that takes $\ket{G}$ to $\ket{G'}$, i.e.
    \begin{equation}
        \ket{G}\xrightarrow[\mathrm{LPM}+\mathrm{CC}]{\mathrm{LC}}\ket{G'}\otimes\ket{\mathrm{junk}}_{V\setminus U}.
    \end{equation}
    If $\ket{G'}$ is a qubit-minor of $\ket{G}$, we denote this as
    \begin{equation}
        \ket{G'}<\ket{G}.
    \end{equation}
\end{definition}
The reason for calling $\ket{G'}$ a qubit-minor of $\ket{G}$ is that this question is in fact equivalent to the graph theoretical question of whether $G'$ is a vertex-minor of $G$, as we show in theorem~\ref{thm:QMVM}.
In the graph theoretical picture, single-qubit Cliffords operations will be replaced by operations called \emph{local complementation}s on the graph and single-qubit Pauli measurements by local complementations and vertex-deletions.
A vertex-minor of some graph is by definition a graph that can be reached by some sequence of local complementations and vertex-deletions.
Equivalence of graphs under local complementations has been studied by Bouchet in~\cite{Bouchet1991}, which was used by Van den Nest et al. in~\cite{VandenNest2004a} to find an efficient algorithm to decide whether two graph states are equivalent under single-qubit Clifford operations.

The computational complexity of deciding if a graph $G'$ is a vertex-minor of $G$, and therefore if $\ket{G'}$ is a qubit-minor of $\ket{G}$, was, to the authors' knowledge, previously unknown.
In another paper~\cite{npcomplete} we show that this decision problem is in fact $\mathbb{NP}$-Complete.
There is therefore no efficient algorithm that solves this question in general, unless $\mathbb{P}=\mathbb{NP}$.

\subsection{Results}
Here, we show that the same problem can be solved in cubic time in the number of qubits of $\ket{G}$ on instances where the \emph{Schmidt-rank width} of $\ket{G}$ and the number of qubits of $\ket{G'}$ are bounded\footnote{Note that the time-dependence on the size of $G'$ can be removed if conjecture~\ref{conj:qrFPT} is true.}.
This is our first main result which we formally state in theorem~\ref{thm:QM_FPT} and prove in section~\ref{sec:courcelle}\ref{subsec:vm_formula}.
\begin{theorem}\label{thm:QM_FPT}
    There exists an algorithm that decides if $\ket{G'}$ is a qubit-minor of $\ket{G}$, and therefore if $G'$ is a vertex-minor of $G$, and has running time
    \begin{equation}
        \mathcal{O}(f(\abs{G'},r)\cdot \abs{G}^3),
    \end{equation}
    where $r$ is the rank-width of $G$ which is equal to the Schmidt-rank width of $\ket{G}$, $\abs{G}$ denotes the number of vertices in the graph $G$ and $f$ is some computable function.
    If conjecture~\ref{conj:qrFPT} is true then there exists an algorithm to the same problem but with running time
    \begin{equation}
        \mathcal{O}(f(r)\cdot\abs{G}^3).
    \end{equation}
\end{theorem}
Our second main result concerns $\mathrm{GHZ}$-states, which are useful for many applications on a quantum network and is therefore an important target state.
A $\mathrm{GHZ}$-state on the qubits in the set $U$ is given as
\begin{equation}
    \ket{\mathrm{GHZ}}_U=\frac{1}{\sqrt{2}}\left(\bigotimes_{v\in U}\ket{0}_v+\bigotimes_{v\in U}\ket{1}_v\right).
\end{equation}
One can easily check that the state $\ket{\mathrm{GHZ}}_U$ is single-qubit Clifford equivalent to the graph state $\ket{K_U}$, where $K_U$ is the complete graph with vertex-set $U$.
It is therefore the case that $\ket{\mathrm{GHZ}}_U$ can be mapped from $\ket{G}$ by $\mathrm{LC}+\mathrm{LPM}+\mathrm{CC}$ if and only if $\ket{K_U}$ is a qubit-minor of $\ket{G}$.
We show that this question can be solved efficiently if the Schmidt-rank width of $\ket{G}$ is bounded, as captured in theorem~\ref{thm:GHZ_FPT} and proven in section~\ref{sec:courcelle}\ref{subsec:vm_formula}.
\begin{theorem}\label{thm:GHZ_FPT}
    There exists an algorithm that decides if $\ket{K_U}$ is a qubit-minor of $\ket{G}$, and therefore if $K_U$ is a vertex-minor of $G$ and has running time 
    \begin{equation}
        \mathcal{O}(f(r)\cdot \abs{G}^3),
    \end{equation}
    where $r$ is the rank-width of $G$ which is equal to the Schmidt-rank width of $\ket{G}$, $\abs{G}$ denotes the number of vertices in the graph $G$ and $f$ is some computable function.
\end{theorem}
Note in particular that the running time in theorem~\ref{thm:GHZ_FPT} does not depend on $U$, even if conjecture~\ref{conj:qrFPT} is false.
Similarly to theorem~\ref{thm:GHZ_FPT} one can also decide if a graph state has a qubit-minor on a subset $U$ with a given property\footnote{Expressible in C$_2$MS.}, efficiently on graph states with bounded Schmidt-rank width, see theorem~\ref{thm:prop_VM}.

Both of the two main results, theorem~\ref{thm:QM_FPT} and theorem~\ref{thm:GHZ_FPT}, rely on a variant of Courcelle's theorem, which we describe more in detail in section~\ref{sec:courcelle}.
Courcelle's theorem states that a large class of graph problems are fixed-parameter tractable.
This means that there exist algorithms for these problems which are efficient in the size of the input graphs, provided a certain parameter of these graphs is bounded.
This is a very powerful theorem, but a direct implementation of the algorithm given by Courcelle's theorem is not useful in practice.
The reason being that even though the algorithm is efficient, the hidden constant factor of the algorithm's asymptotic runtime is huge.
This huge constant factor is unavoidable since the theorem is so general and captures many $\mathbb{NP}$-Complete problems.
On the other hand, by knowing that a problem can be efficiently solved, one can usually find a more tailored efficient algorithm for the problem at hand, that does not have a huge hidden constant in the runtime.
In another paper~\cite{npcomplete}, we provide an efficient algorithm without a huge hidden constant for the problem of deciding whether $\ket{K_U}$ is a qubit-minor of some graph $\ket{G}$, if $\ket{G}$ has Schmidt-rank width one.
There are also many other approaches to find practical algorithms for problems captured by Courcelle's theorem, see for example~\cite{Langer2014} or~\cite{Ganian2010}.


We have implemented many of the concepts and algorithms mentioned in this paper in SAGE~\cite{sage} and MONA~\cite{mona}.
Both the code in SAGE and MONA can be freely accessed from the git-repository at~\cite{git}.
The functionalities provided by this repository include:
\begin{itemize}
    \item A function taking two graphs, $G$ and $G'$, as input and returns $\textproc{True}$ if the graph states $\ket{G}$ and $\ket{G'}$ are equivalent under single-qubit Clifford operations and otherwise returns $\textproc{False}$. The function has a runtime of $\mathcal{O}(\abs{G}^4)$ and is an implementation of the algorithm described in~\cite{Bouchet1991,VandenNest2004a}.
    \item A function taking two graphs, $G$ and $G'$, as input and returns a sequence of operations that takes $\ket{G}$ to a graph state which is single-qubit Clifford equivalent to $\ket{G'}$, if $\ket{G'}<\ket{G}$ and otherwise returns $\textproc{False}$. This function uses a more sophisticated version of the non-efficient algorithm described in section~\ref{sec:background}\ref{sec:brute}.
    \item A function taking a graph $G$ and a set $U$ as input and either returns a sequence of operations that takes $\ket{G}$ to a graph state which is single-qubit Clifford equivalent to $\ket{K_U}$ and therefore $\ket{\mathrm{GHZ}}_U$ or returns $\textproc{False}$. If the function returns $\textproc{False}$ and $G$ has rank-width one, then $\ket{K_U}\nless\ket{G}$ as we show on~\cite{npcomplete}. The runtime of this function is $\mathcal{O}(\abs{U}\abs{G}^3)$.
    \item As described in section~\ref{sec:courcelle} one can express whether $\ket{G'}<\ket{G}$ in a logic called monadic second-order logic. We have implemented the expression for $\ket{G'}<\ket{G}$ in MONA which is a software to translate such logic expressions to finite-state automata. This can then be used to construct efficient algorithms for graphs of bounded rank-width.
\end{itemize}

\subsection{Overview}
In section~\ref{sec:background} we describe graph states and introduce the notions of qubit-minors and vertex-minors.
We also provide a non-efficient but correct algorithm for deciding the vertex-minor problem, taking any graph as input, in section~\ref{sec:brute}.
Note that such an algorithm is necessarily non-efficient, unless $\mathbb{P}=\mathbb{NP}$, since the problem it solves is in general $\mathbb{NP}$-Complete.
Furthermore, we describe how the corresponding operations on the graph states can be applied in constant time in section~\ref{sec:constant}.
In section~\ref{sec:courcelle} we provide an efficient algorithm for graphs with bounded rank-width by making use of monadic second-order logic and Courcelle's theorem.
It is our intention that this section can also be used as a short introduction for those not familiar with these concepts.

\section{Background}\label{sec:background}

\subsection{Notation}
All graphs in this paper are simple, unweighted and undirected, where simple means that there are no self-loops or multi-edges.
Given a graph $G=(V,E)$, we will sometimes denote the vertex-set as $V(G)=V$ and the edge-set as $E(G)=E$.
By the size of a graph we mean the number of vertices, which we denote as $\abs{G}=\abs{V(G)}$.
We will denote the neighborhood of a vertex as
\begin{equation}
    N_v^{(G)}=\{u\in V(G)\;:\;(v,u)\in E(G)\}.
\end{equation}
If it is clear which graph is considered, we will also sometimes write $N_v$.
The induced subgraph of $G$ on the subset $U\subseteq V(G)$ is denoted as $G[U]$ and is the graph with vertex-set $U$ and edge-set
\begin{equation}
    \{(u,v)\in E(G):u\in U\land v\in U\}.
\end{equation}
We denote vertex-deletions by $\setminus v$ such that $G\setminus v=G[V(G)\setminus \{v\}]$.

The Pauli matrices will be denoted as
\begin{equation}
    \mathbb{I}=\begin{pmatrix}1 & 0\\0 & 1\end{pmatrix},\quad X=\begin{pmatrix}0 & 1\\1 & 0\end{pmatrix},\quad Y=\begin{pmatrix}0 & -\mathrm{i}\\\mathrm{i} & 0\end{pmatrix},\quad Z=\begin{pmatrix}1 & 0\\0 & -1\end{pmatrix}.
\end{equation}
The single-qubit Clifford group $\mathcal{C}$ is the normalizer of the Pauli group $\mathcal{P}=\langle \mathrm{i}\mathbb{I},X,Z\rangle$, i.e.
\begin{equation}
    \mathcal{C}=\big\{C\in \mathcal{U}:(\forall P\in\mathcal{P}:CPC^\dagger\in\mathcal{P})\big\},
\end{equation}
where $\mathcal{U}$ is the single-qubit unitary operations.

Assume that $v_i$ is the label of a qubit which is part of some multi-qubit state $\ket{\psi}_{v_1\dots v_i \dots v_n}$.
We will then denote $P_{v_i}$ as the operation
\begin{equation}
    P^{(v_1\dots v_n)}_{v_i}=(\mathbb{I})_{v_1}\otimes\dots\otimes (P)_{v_i}\otimes\dots\otimes(\mathbb{I})_{v_n},
\end{equation}
where $P\in\{\mathbb{I},X,Y,Z\}$.
We will never write an explicit ordering of the qubits in a multi-qubit state $\ket{\psi}_{v_1\dots v_i \dots v_n}$ and rather write $\ket{\psi}_V$, where $V$ is the set $\{v_1,\dots,v_n\}$.
For explicit calculations one just needs to use a consistent ordering.
Similarly for the operation $P^{(v_1\dots v_n)}_{v_i}$ we will write $P^{V}_{v_i}$ or even $P_{v_i}$ when it is clear which set $V$ is considered.

\subsection{Graph states}\label{sec:graphstates}
A graph state $\ket{G}$ is a quantum state described by a simple unweighted undirected graph $G$, where the vertices of $G$ correspond to the qubits of $\ket{G}$. Formally, let $G=(V,E)$ be a graph, the graph state $\ket{G}$ is then defined as
\begin{equation}
    \ket{G}=\prod_{e\in E}C_Z^{e}\bigotimes_{v\in V}\ket{+}_v,
\end{equation}
where $C_Z^{(u,v)}$ is a controlled phase gate between qubits $u$ and $v$, i.e.
\begin{equation}
    C_Z^{(u,v)}=\ket{0}\bra{0}_u\otimes\mathbb{I}^{(V\setminus u)}_v+\ket{1}\bra{1}_u\otimes Z^{(V\setminus u)}_v.
\end{equation}
A graph state is also a stabilizer state~\cite{Hein2006}.
The generators of the stabilizer group of $\ket{G}$ can be written as
\begin{equation}
    g_v=X_v\prod_{u\in N_v}Z_u.
\end{equation}
In fact, any stabilizer state can be made into some graph state by only performing single-qubit Clifford operations~\cite{VandenNest2004}.
If two states, $\ket{\psi}$ and $\ket{\phi}$ are related by some sequence of single-qubit Clifford operations, we denote this by $\ket{\psi}\sim_\mathrm{LC}\ket{\phi}$.
Given a stabilizer state $\ket{S}$, one can find a graph state $\ket{G}$ such that $\ket{S}_V\sim_{\mathrm{LC}}\ket{G}_V$, by simply performing Gaussian elimination followed by certain operations on the columns of a matrix which rows are the symplectic form of the generators of $S$~\cite{Hein2006}.

To study what graph states can be reached from a given graph state we introduce the notion of a qubit-minor as defined in definition~\ref{def:QM}.
It turns out that the question of whether $\ket{G'}$ is a qubit-minor of $\ket{G}$ is equivalent to whether $G'$ is a vertex-minor of $G$, as we describe below.

\subsection{Local Clifford operations}\label{subsec:cliffords}
Let's consider the following sequence of single-qubit Clifford operations
\begin{equation}\label{eq:LCU}
    U_v^{(G)}=\exp\left(-\mathrm{i}\frac{\pi}{4}X_v\right)\prod_{u\in N_v}\exp\left(\mathrm{i}\frac{\pi}{4}Z_u\right).
\end{equation}
As shown in~\cite{Hein2006}, the operation $U_v^{(G)}$ on the state $\ket{G}$ can be seen as an operation on the graph $G$ since
\begin{equation}
    U_v^{(G)}\ket{G}=\ket{\tau_v(G)}
\end{equation}
where $\tau_v$ is a local complementation on the vertex $v$, as defined in definition~\ref{def:LC} and illustrated in equation~\eqref{eq:vis_LC}.

\begin{definition}[local complementation]\label{def:LC}
    A local complementation $\tau_v$ acts on a vertex $v$ of a graph $G$ by complementing the induced subgraph on the neighborhood of $v$. The neighborhoods of the graph $\tau_v(G)$ are therefore given by
    \begin{equation}
    N_u^{(\tau_v(G))}=\begin{cases}N_u\Delta (N_v\setminus\{u\}) & \quad \text{if } (u,v)\in E(G) \\ N_u & \quad \text{else}\end{cases},
    \end{equation}
    where $\Delta$ denotes the symmetric difference between two sets.
\end{definition}

\begin{equation}\label{eq:vis_LC}
\raisebox{-0.08\textwidth}{\includegraphics[scale=0.5]{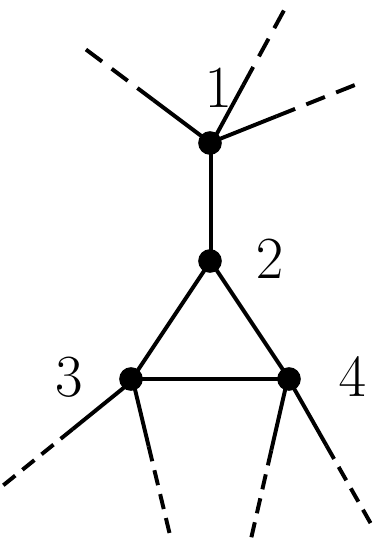}}\quad\xrightarrow{\tau_2}\quad\raisebox{-0.08\textwidth}{\includegraphics[scale=0.5]{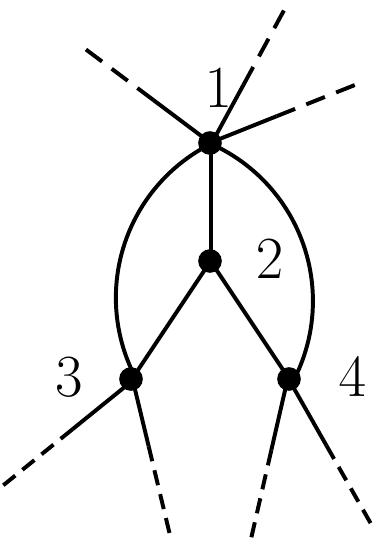}}
\end{equation}

Surprisingly, any single-qubit Clifford operation which takes some graph state to another graph state can seen as some sequence of local complementations on the corresponding graphs.
This was proven in~\cite{VandenNest2004} and we restate this theorem here.
\begin{theorem}[Van den Nest \cite{VandenNest2004}]\label{thm:nest}
    Two graph states $\ket{G}$ and $\ket{G'}$ are equivalent under single-qubit Clifford operations if and only if their corresponding graphs $G$ and $G'$ are related by some sequence of local complementations.
\end{theorem}
Let $m=v_1v_2\dots v_l$ be a sequence of vertices of $G$, then we denote the sequence of local complementations on the vertices in $m$ as
\begin{equation}
    \tau_m(G)=\tau_{v_l}\circ\dots\circ\tau_{v_2}\circ\tau_{v_1}(G).
\end{equation}
If there exists a sequence $m$ such that $\tau_m(G)=G'$ then we write this as $G\sim_\mathrm{LC}G'$.
Theorem~\ref{thm:nest} can therefore be stated as
\begin{equation}
    \ket{G}\sim_\mathrm{LC}\ket{G'}\quad\Leftrightarrow\quad G\sim_\mathrm{LC}G'.
\end{equation}

Testing whether two graphs are LC-equivalent can be done in time $\mathcal{O}(n^4)$, where $n$ is the size of the graphs, as shown in~\cite{Bouchet1991}.


\subsection{Local Pauli measurements}\label{subsec:meas}
How should the corresponding graph of a graph state be updated when a measurement is performed?
In~\cite{Hein2004} it is shown how the Pauli projectors $P_v^{(X,\pm)}$, $P_v^{(Y,\pm)}$, $P_v^{(Z,\pm)}$ act on graph states\footnote{For the special case when $\abs{N_v}=0$, a measurement in the $X$-basis does not change the graph state since this is then $\ket{G}=\ket{+}_v\otimes\ket{G\setminus v}$.}
\begin{align}
    P_v^{(Z,\pm)}\ket{G} &= \frac{1}{2}\ket{Z,\pm}_v\otimes U_v^{(Z,\pm)}\ket{G\setminus v} \label{eq:meas_Z}\\
    P_v^{(Y,\pm)}\ket{G} &= \frac{1}{2}\ket{Y,\pm}_v\otimes U_v^{(Y,\pm)}\ket{\tau_v(G)\setminus v} \label{eq:meas_Y}\\
    P_v^{(X,\pm)}\ket{G} &= \frac{1}{2}\ket{X,\pm}_v\otimes U_e^{(X,\pm)}\ket{T_{e}(G)\setminus v} \quad\text{ if }\abs{N_v}>0\label{eq:meas_X}
\end{align}
where $e$ is an edge of $G$ incident on the vertex $v$.
Choosing a different edge $e'$ incident on $v$ gives a single-qubit Clifford-equivalent graph state, i.e. $\ket{T_{e'}(G)\setminus v}\sim_\mathrm{LC}\ket{T_e(G)\setminus v}$.
The operation $T_{(u,v)}$ is called a pivot and is defined as $T_{(u,v)}=\tau_v\circ\tau_u\circ\tau_v$.
The pivot can simply be specified by an undirected edge since
\begin{equation}
    \tau_v\circ\tau_u\circ\tau_v(G)=\tau_u\circ\tau_v\circ\tau_u(G)\quad \text{if }(u,v)\in E(G)
\end{equation}
as shown in~\cite{Bouchet1988a}.
The operators $U_v^{(P,\pm)}$ are sequences of single-qubit Clifford operations and take the post-measurement state to a graph state.
The exact form of these correction operators can be found in~\cite{Hein2004} but we will only need the ones for measurements in the standard ($Z$-) basis, which are given by
\begin{equation}\label{eq:corr}
    U_v^{(Z,+)}=\mathbb{I}_v,\quad U_v^{(Z,-)}=\prod_{u\in N_v}Z_u.
\end{equation}

Since the correction operators are sequences of single-qubit Clifford operations it is therefore the case that
\begin{equation}\label{eq:qminors}
    \ket{G\setminus v},\quad \ket{\tau_v(G)\setminus v},\quad \ket{T_e(G)\setminus v}
\end{equation}
are all qubit-minors of $\ket{G}$.
Furthermore, Bouchet proved the following.
\begin{lemma}[Bouchet, (9.2) in~\cite{Bouchet1988a}]\label{lem:single_vertex-minor}
    If $G\sim_\mathrm{LC}G'$ then $G'\setminus v$ is LC-equivalent to $G\setminus v$, $\tau_v(G)\setminus v$ or $T_e(G)\setminus v$, where $e$ is some fixed edge incident on $v$ in $G$.
\end{lemma}
We therefore have the following lemma.
\begin{lemma}\label{lem:single_qubit-minor}
    Let $\ket{G}$ be a graph state, $v\in V(G)$ be a vertex and $e\in E(G)$ be an edge incident on $v$.
    Furthermore, assume that $\ket{G'}$ is a qubit-minor of $\ket{G}$, where $V(G')=V(G)\setminus v$ and that $G'$ has no vertices of degree zero.
    Then $\ket{G'}$ is single-qubit Clifford-equivalent to at least one of the three states in equation~\eqref{eq:qminors}.
\end{lemma}
\begin{proof}
    Since $\ket{G'}$ is a qubit-minor of $\ket{G}$ we know that there exists a sequence $\mathcal{W}$ of single-qubit Clifford operations, single-qubit Pauli measurements and classical communication that takes $\ket{G}$ to $\ket{G'}$, by definition.
    Any single-qubit Pauli measurement on a qubit $u$ gives a product states between qubit $u$ and the rest of the qubits, see equations~\eqref{eq:meas_Z}-\eqref{eq:meas_X}.
    The sequence $\mathcal{W}$ cannot therefore contain a single-qubit Pauli measurement on a qubit $u$, different from $v$, since $u$ has non-zero degree in the graph $G'$ by assumption.
    Without loss of generality we can in fact assume that $\mathcal{W}$ is a sequence of single-qubit Clifford operations followed by a measurement of $v$ in the standard basis and then by another sequence of single-qubit Clifford operations.
    The reason why it is sufficient to consider a measurement in the standard basis is that any other Pauli measurement can be simulated by performing some single-qubit Clifford operation followed by a measurement in the standard basis.
    Assume that the sequence of single-qubit Clifford operations before (after) the measurement is described by the sequences of local complementations $m$ ($m'$), which exists due to theorem~\ref{thm:nest}.
    We therefore have that
    \begin{equation}
        \tau_{m'}(\tau_m(G)\setminus v)=G'.
    \end{equation}
    Using lemma~\ref{lem:single_vertex-minor} we have that $\tau_m(G)\setminus v$, and therefore $G'$, is LC-equivalent to either $G\setminus v$, $\tau_v(G)\setminus v$ or $T_e(G)\setminus v$.
    Finally, by theorem~\ref{thm:nest} the lemma follows.
\end{proof}

\subsection{Vertex-minors}
As mentioned, the question of whether $\ket{G'}$ is a qubit-minor of $\ket{G}$ is equivalent to whether $G'$ is a vertex-minor of $G$.
Vertex-minors were introduced by Bouchet in~\cite{Bouchet1988a} but by the name of $l$-reductions.
\begin{definition}[vertex-minor]
    Let $G$ be a graph. $G'$ is called a vertex-minor of $G$ if it can be reached by some sequence of local complementations and vertex-deletions.
    Equivalently, $G'$ is called a vertex-minor of $G$ if there exists a sequence of vertices $m$ such that
    \begin{equation}\label{eq:LC_VM}
        \tau_m(G)[V(G')]=G'
    \end{equation}
    If $G'$ is a vertex-minor of $G$, we denote this as
    \begin{equation}
        G'<G
    \end{equation}
\end{definition}
In the previous sections we have seen that single-qubit Clifford operations that take graph states to graph states can be seen as local complementations on the corresponding graph and similarly for single-qubit Pauli measurements and vertex-deletions.
The relation between qubit-minors and vertex-minors is captured by the following theorem.
\begin{theorem}\label{thm:QMVM}
    Let $\ket{G}$ and $\ket{G'}$ be two graph states such that no vertex in $G'$ has degree zero.
    $\ket{G'}$ is then a qubit-minor of $\ket{G}$ if and only if $G'$ is a vertex-minor of $G$, i.e.
    \begin{equation}
        \ket{G'}<\ket{G}\quad\Leftrightarrow\quad G'<G.
    \end{equation}
\end{theorem}
\begin{proof}
    Assume first that $\ket{G'}$ is a qubit-minor of $\ket{G}$.
    By the same arguments as in the proof of lemma~\ref{lem:single_qubit-minor} we then have that there exists a sequence of vertices $m$ such that
    \begin{equation}
        \tau_m(G)[V(G')]=G'
    \end{equation}
    and by definition that $G'$ is a vertex-minor of $G$.
    Assume now on the other hand that $G'$ is a vertex-minor of $G$, i.e. that $\tau_m(G)[V(G')]=G'$ for some $m$.
    We can then go from $\ket{G}$ to $\ket{G'}$ by simply performing the single-qubit Clifford operations corresponding to the sequence of local complementation specified by $m$ and then measure the qubits $V(G)\setminus V(G')$ in the standard basis.
    If the correct corrections, i.e. $U_v^{(Z,\pm)}$, are applied after the measurements, the state $\ket{G'}$ is reached.
\end{proof}
So to check whether a graph state has a certain qubit-minor we can check if the corresponding graph has a certain vertex-minor.
Note that one can also include the case where $G'$ has vertices of degree zero.
Let's denote the vertices of $G'$ which have degree zero as $U$.
We then have that
\begin{equation}
    \ket{G'}<\ket{G}\quad\Leftrightarrow\quad G'[V(G)\setminus U]<G.
\end{equation}

\subsection{Rank-width}\label{app:rwd}
In this section we introduce the notion of rank-width, which is a complexity measure of a graph.
It is in some ways similar to the tree-width, introduced in~\cite{Robertson1986}.
The tree-width captures essentially how tree-like the graph is.
This is useful for finding algorithms for problems on graphs of bounded tree-width, motivated by the fact that many graph problems are easy on trees.
More on algorithms for problems on graphs of bounded tree-width can be found in~\cite{Downey1999}.
Rank-width, compared to tree-width, captures a larger class of graphs with similar complexity.
For example, the complete graph has very low complexity, due to its highly symmetric nature, but the tree-width is in this case maximal.
On the other hand the rank-width is one for both trees and complete graphs.
In fact, it turns out that the graphs of rank-width one are exactly the distance-hereditary graphs, see~\cite{Oum2005}.

We start by defining the cut-rank of a graph.
To do this we will use the following notation for a graph $G$ with vertices $V$ and adjacency matrix $\Gamma$ and two subsets of the vertices $A,B\subseteq V$; $\Gamma[A,B]$ is the $\abs{A}\times\abs{B}$-matrix describing the connections between the sets $A$ and $B$.
So, for $a\in A$ and $b\in B$, the element $(\Gamma[A,B])_{ab}$ is $1$ if $(a,b)$ is an edge in $G$ and $0$ otherwise.
\begin{definition}[cut-rank]
    Let's assume that $A$ is a subset of the vertices $V$ of some graph $G$ with adjacency matrix $\Gamma$.
    The cut-rank $\cutrk_A(G)$ of $G$ with respect to $A$, is then defined as
    \begin{equation}
        \cutrk_A(G)\equiv\rank_{\mathbb{F}_2}(\Gamma[A,V\setminus A]),
    \end{equation}
    where $\rank_{\mathbb{F}_2}$ is the rank over the finite field of order two.
\end{definition}
Note that the cut-rank is symmetric in the sense that
\begin{equation}
    \cutrk_A(G)=\rank_{\mathbb{F}_2}(\Gamma[A,V\setminus A])=\rank_{\mathbb{F}_2}(\Gamma[V\setminus A,A]^\top)=\cutrk_{V\setminus A}(G).
\end{equation}
Interestingly the cut-rank with respect to $A$ of a graph $G$ is in fact equal to the Schmidt-rank of the state $\ket{G}$ with respect to the bipartition $(A,V\setminus A)$.\footnote{This is proven in~\cite{Hein2006}, see proposition 10.}

Next we define what is called a rank-decomposition of a graph.
\begin{definition}[rank-decomposition]
    A rank-decomposition of a graph $G$ is a pair $\mathcal{R}=(\mathcal{T},\mu)$, where $\mathcal{T}$ is a subcubic tree and $\mu$ is a bijection $\mu:V(G)\rightarrow\{l:l\text{ is a leaf of }\mathcal{T}\}$.
    A subcubic tree is a tree with at least two vertices and each vertex has degree less or equal to $3$.
    Any edge $e$ in $\mathcal{T}$ splits the tree into two connected components upon deletion and therefore induces a partition $(A_e,B_e)$ of the leaves.
    The width of an edge $e$ of the subcubic tree is defined as the cut-rank of the corresponding partition.
    Furthermore the width of the rank-decomposition is defined as the maximum width over all edges, i.e.
    \begin{equation}
        \width_\mathcal{R}(G)\equiv\max_{e\in E(\mathcal{T})}\cutrk_{\mu^{-1}(A_e)}(G).
    \end{equation}
    To simplify notation we write the cut-rank induced by a rank-decomposition $(\mathcal{T},\mu)$ and an edge $e$ as
    \begin{equation}
        \cutrk_{\mu^{-1}(\mathcal{T},e)}(G)\equiv \cutrk_{\mu^{-1}(A_e)}(G).
    \end{equation}
\end{definition}
This allows us to define the rank-width of a graph.
\begin{definition}[rank-width]\label{def:rankwidth}
    The rank-width $\rwd(G)$ of a graph $G$ is the minimum width over all rank-decompositions, i.e.
    \begin{equation}
        \rwd(G)\equiv\min_\mathcal{R}\width_\mathcal{R}(G)=\min_{(\mathcal{T},\mu)} \max_{e\in E(\mathcal{T})}\cutrk_{\mu^{-1}(\mathcal{T},e)}(G).
    \end{equation}
\end{definition}
The rank-width of the graph $G$ is related to the entanglement of the state $\ket{G}$, although as a relatively unknown entanglement monotone.
In~\cite{VandenNest2007} the corresponding entanglement monotone is called the \textit{Schmidt-rank width} and is defined for general quantum states.
For graph states, the Schmidt-rank width of the state and the rank-width of the corresponding graph coincide.
There they also give an interpretation of the Schmidt-rank width as a quantifier for the optimal description of the state using a tree tensor network.

\section{A non-efficient but general algorithm}\label{sec:brute}
Here we describe an algorithm that decides if $G'$ is a vertex-minor of $G$ and returns a sequence of local complementations $m$ such that $\tau_m(G)[V(G')]\sim_\mathrm{LC}G'$ if such a sequence exists.
This algorithm works for any $G$ and $G'$ and has a running time of $\mathcal{O}(3^{n-k}(k^4+(n-k)n^2)$, where $k=\abs{G'}$ and $n=\abs{G}$.
Obviously this algorithm is not efficient, due to the exponential scaling in the size-difference of the graphs, but is still useful for smaller graphs or when $n-k$ is bounded.
For finite $n$, the algorithm described in this section is also a useful benchmark for the efficient algorithm described in section~\ref{sec:courcelle} for graphs of bounded rank-width.
To prove that the algorithm is correct, we first prove theorem~\ref{thm:multi_vertex-minor}.

\begin{algorithm}[H]
    \caption{Non-efficient algorithm that decides if $G'<G$. \newline Input: ($G$,$G'$). \newline Output: A sequence $m$ such that $\tau_m(G)[V(G')]\sim_\mathrm{LC}G'$ if $G'<G$.\newline \hbox{}\hspace{1.2cm}$\perp$ \hspace{6.23cm}if $G'\nless G$.}\label{alg:vertex-minor}
    \begin{algorithmic}[1]
        \Function{is\_vm}{$G$,$G'$}
            \If{$V(G')\nsubseteq V(G)$}
                \State \textbf{return} $\perp$
            \EndIf
            \If{$V(G)=V(G')$}
                \If{$G\sim_\mathrm{LC}G'$}
                    \State \textbf{return} []\Comment{Return an empty sequence}
                \Else
                    \State \textbf{return} $\perp$
                \EndIf
            \Else
                \State Let $v$ be a vertex in $V(G)\setminus V(G')$
                \State Let $m_Z=\textproc{is\_vm}(G\setminus v)$
                \If{$m_Z\neq\perp$}
                    \State \textbf{return} $m_Z$
                \Else
                    \State Let $m_Y=\textproc{is\_vm}(\tau_v(G)\setminus v)$
                    \If{$m_Y\neq\perp$}
                        \State \textbf{return} $m_Y\parallel [v]$\Comment{Concatenate $[v]$ to $m_Y$ and return.}
                    \Else
                        \State Let $u$ be a vertex incident to $v$ in $G$
                        \State Let $m_X=\textproc{is\_vm}(T_{(v,u)}(G)\setminus v)$
                        \If{$m_X\neq\perp$}
                            \State \textbf{return} $m_X\parallel [u,v,u]$\Comment{Concatenate $[u,v,u]$ to $m_X$ and return.}
                        \Else
                            \State \textbf{return} $\perp$
                        \EndIf
                    \EndIf
                \EndIf
            \EndIf
        \EndFunction
    \end{algorithmic}
\end{algorithm}

As seen in the previous section, when making a Pauli $X$ measurement on a qubit $v$, the choice of the edge $e$ incident on $v$ gives different graph states that are single-qubit Clifford equivalent to the post-measurement state.
To simplify the results of this section we introduce the following graph operation which allows us to not have to deal with different choices of edges, incident on the qubit being measured.
\begin{definition}
    Let $v$ be a vertex in the graph $G$.
    We define $T_v$ as the following operation
    \begin{equation}
        T_v(G)=\begin{cases}T_{e_v}(G) & \quad\text{if }\abs{N_v}>0 \\ G & \quad\text{if }\abs{N_v}=0\end{cases}
    \end{equation}
    where $e_v$ is an edge incident on $v$ chosen in some consistent way.
    For example we could assume that the vertices of $G$ are ordered and that $e_v=(v,\min(N_v))$.
    The specific choice will not matter but importantly $e_v$ only depends on $G$ and $v$, and the same therefore holds for $T_v(G)$.
\end{definition}
We are now ready to prove a generalization of lemma~\ref{lem:single_qubit-minor}, which algorithm~\ref{alg:vertex-minor} is built on.
\begin{theorem}\label{thm:multi_vertex-minor}
    Let $G$ and $G'$ be two graphs and $U$ be the set $V(G)\setminus V(G')=\{v_1,\dots,v_{n-k}\}$.
    Furthermore, let $\mathcal{P}_U$ denote the set of graph operations
    \begin{equation}
        \mathcal{P}_U=\{P_{v_{n-k}}\circ\dots\circ P_{v_1}\;:\;P_{v}\in\{({\_\,})\setminus v,\,\tau_v({\_\,})\setminus v,\,T_v({\_\,})\setminus v\}\}
    \end{equation}
    Then we have that
    \begin{equation}
        G'<G\quad\Leftrightarrow\quad \exists P\in\mathcal{P}_U\;:\;G'\sim_\mathrm{LC}P(G).
    \end{equation}
\end{theorem}
\begin{proof}
    If there exists a $P$ in $\mathcal{P}_U$ such that $G'\sim_\mathrm{LC}P(G)$ then we clearly have that $G'<G$, since any such $P$ is some sequence of local complementations and vertex-deletions.
    Assume now that $G'<G$.
    We will prove by induction on $n-k$ that there exists a $P$ in $\mathcal{P}_U$ such that $G'\sim_\mathrm{LC}P(G)$.
    For $n-k=1$ this follows directly from lemma~\ref{lem:single_vertex-minor}.
    Assume therefore that it is true for $n-k=l$.
    We now show that this implies that it is also true for $n-k=l+1$.
    Since $G'<G$ we know that $\tau_m(G)[V(G')]=G'$ for some $m$.
    Let $v$ be a vertex in $V(G')$ and consider the graph $\tilde{G}=\tau_m(G)[V(G')\cup\{v\}]$.
    Note that $G'=\tilde{G}\setminus v$.
    Clearly we have that $G'<\tilde{G}<G$ and by the induction assumption we know that
    \begin{equation}
        \exists P\in \mathcal{P}_{U\setminus \{v\}}\;:\;\tilde{G}\sim_\mathrm{LC}P(G).
    \end{equation}
    Then from lemma~\ref{lem:single_vertex-minor} we know that $G'=\tilde{G}\setminus v$ is LC-equivalent to at least one of the following graphs
    \begin{equation}
        P(G)\setminus v,\quad \tau_v(P(G))\setminus v,\quad T_v(P(G))\setminus v
    \end{equation}
    and the theorem follows.
\end{proof}
From theorem~\ref{thm:multi_vertex-minor} we see that to check if $G'$ is a vertex-minor of $G$ it is sufficient to check if one of the graphs in
\begin{equation}\label{eq:set_candidates}
    \left\{P(G)\;:\;P\in\mathcal{P}_{V(G)\setminus V(G')}\right\}
\end{equation}
are LC-equivalent to $G'$.
Note that there are possibly $3^{\abs{V(G)\setminus V(G')}}$ graphs in the set in equation~\eqref{eq:set_candidates} to check.
As mentioned earlier it is possible to check whether two graphs are LC-equivalent in time $\mathrm{O}(k^4)$, where $k$ is the size of the graphs.
The explicit algorithm for checking if $G'$ is a vertex-minor of $G$ is stated in algorithm~\ref{alg:vertex-minor} and theorem~\ref{thm:alg_VM} captures the proof that it is correct and what its running time is.

\begin{theorem}\label{thm:alg_VM}
    Algorithm~\ref{alg:vertex-minor} returns a sequence $m$ such that $\tau_m(G)[V(G')]\sim_\mathrm{LC}G'$ if $G'<G$ and returns $\perp$ if $G'\nless G$.
    Furthermore the runtime is $\mathcal{O}(3^{n-k}(k^4+(n-k)n^2)$, where $k=\abs{G'}$ and $n=\abs{G}$.
\end{theorem}
\begin{proof}
    We first prove that the algorithm is correct.
    Since the algorithm calls itself recursively with the three graphs as in in line 13, 17 and 21 it will generate all the graphs in equation~\eqref{eq:set_candidates}.
    That is, the graphs tested for LC-equivalence against $G'$ in line 6 are exactly the graphs in equation~\eqref{eq:set_candidates}.
    Furthermore, if at least one of the base-cases, i.e. line 5-10, return an empty sequence then the top-level call to the algorithm will return a sequence $m$ such that $\tau_m(G)[V(G')]\sim_\mathrm{LC}G'$.
    On the other hand, if all of the base-cases returns $\perp$ then the top-level call returns $\perp$.
    By theorem~\ref{thm:multi_vertex-minor} it follows that the algorithm is correct.

    Let's now consider the runtime of the algorithm.
    We assume that the graphs are given as their adjacency matrices.
    Let's denote the running time by $T(n,k)$.
    Picking a vertex $v$ from the set $V(G)\setminus V(G')$, as in line 12, can be done in time $\mathcal{O}(n)$.
    There are three recursive calls on line 13, 17 and 21, where also four local complementations are performed\footnote{Three local complementations for the pivot.}.
    Each local complementation can be done in quadratic time in the size of the graph.
    From the definition of local complementation, definition~\ref{def:LC}, we see that this can be done by adding the row of the adjacency matrix corresponding to the vertex where the local complementation is performed to the rows of its neighbors, where adding means vector-addition modulo 2.
    We therefore have that the running time of the full algorithm has the following recursive expression
    \begin{equation}
        T(n,k)=3T(n-1,k)+\mathcal{O}(n^2)
    \end{equation}
    where $T(k,k)=\mathcal{O}(k^4)$ from testing LC-equivalence of the graphs as the base-case.
    By induction we see that the running time can be expressed as
    \begin{equation}
        T(k+l,k)=3^l\mathcal{O}(k^4)+\sum_{i=0}^{l-1}3^i\mathcal{O}((k+l-i)^2)
    \end{equation}
    Evaluating the above expression for $l=n-k$ we get
    \begin{align}
        T(n,k)&=3^{n-k}\big(\mathcal{O}(k^4)+\sum_{i=0}^{n-k-1}\mathcal{O}((n-i)^2)\big)\\
              &=\mathcal{O}\big(3^{n-k}(k^4+(n-k)n^2\big)
    \end{align}
\end{proof}

\section{Constant time transformation}\label{sec:constant}
So far we have considered the task of finding the operations that take some graph state to its qubit-minor, but what is the best way to apply these operations to the state when they are found?
Assume that we have, by some classical (or quantum) algorithm, found a sequence of operations that takes us from the current state $\ket{G}$ to the target state $\ket{G'}$, i.e. $\ket{G'}$ is a qubit-minor of $\ket{G}$.
There are different ways to express these operations, for example as a sequence of single-qubit Clifford operations and single-qubit Pauli measurements or as a sequence of local complementations and vertex-deletions on the corresponding graph.
From the previous section we have also seen how these different representations can be mapped to each other.
Let's therefore assume that we have expressed the sequence of operations as local complementations $m$ followed by vertex-deletions of the vertices in $V(G)\setminus V(G')$.
The reason for doing this is that we can now perform the single-qubit Clifford operations corresponding to $m$ in parallel and then simultaneously measure all the qubits in $V(G)\setminus V(G')$ in the standard basis.
The simultaneous measurements in the standard basis are possible since the corrections $U_v^{(Z,\pm)}$ are either the identity or $Z$ on the neighbors of $v$ and do therefore not change the measurement basis of neighboring vertices, in contrast to Pauli $X$ and $Y$ measurement.

We still need to know what corrections that are needed, depending on the measurement outcomes of the qubits $V(G)\setminus V(G')$.
In appendix~\ref{app:correct} we show that vertex $v\in V(G')$ only need to apply $(Z_v)^{y_v}$, where
\begin{equation}
    y_v=\sum_{u\in N_v^{(G)}\setminus V(G')}x_u\pmod{2}
\end{equation}
and $x_u\in\{0,1\}$ is the measurement outcome of node $u$.
In other words, a $Z$ operations is applied to qubit $v$ if the parity of the measurement outcomes of the neighbors of $v$ (in $G$) is odd.
Otherwise, no correction is applied to qubit $v$.
We emphasize that the corrections of the qubit $v$ only depend on the measurement outcomes of the neighborhood of that qubit in the graph $G$.


Another advantage of only performing measurements in the standard basis is that in some cases it is possible to extract $\ket{G'}$ without destroying all the rest of the entanglement in the original state.
More specifically, consider the vertices that are adjacent to at least one vertex in $V(G')$ but which are not in $V(G')$ themselves.
These vertices are exactly the ones in the set
\begin{equation}
    N_{V(G')}=\left(\bigcup_{v\in V(G')}N_v^{(\tau_m(G))}\right)\setminus V(G').
\end{equation}
Assume that $N_{V(G')}\neq V(G)\setminus V(G')$.
Then the deletion of all the vertices in $N_{V(G')}$ from the graph $\tau_m(G)$ gives a graph with two connected components $G'=\tau_m(G)[V(G')]$ and $\tau_m(G)[V(G)\setminus (N_{V(G')}\cup V(G'))]$.
Let's denote the second connected component $G_\mathrm{rest}$.
We can then see that if after performing the single-qubit Clifford operations corresponding to $m$, we only measure the qubits in $N_{V(G')}$ in the standard basis, followed by corrections, we arrive at the following state
\begin{equation}
    \ket{G'}\otimes\ket{G_\mathrm{rest}}\otimes\bigotimes_{v\in N_{V(G')}}\ket{+}_v.
\end{equation}
The entanglement in $\ket{G_\mathrm{rest}}$ is then not wasted.
Since there are in fact multiple sequences of vertices $m$ such that $\tau_m(G)[V(G')]=G'$, one can also try to minimize the neighborhood $N_{V(G')}$ and therefore maximize the size of $\ket{G_\mathrm{rest}}$.

\section{Efficient algorithm based on theorem by Courcelle}\label{sec:courcelle}
As mentioned, we show in another paper~\cite{npcomplete} that the problem of deciding if $G'$ is a vertex-minor of $G$ is $\mathbb{NP}$-Complete in general.
Fortunately the problem is fixed-parameter tractable in the rank-width of $G$ and in general the size of $G'$, which follows from results by Oum and Courcelle in ~\cite{Courcelle2007} as we show below.
The statement that a problem is fixed parameter-tractable in some parameter $r$ means that there exists an algorithm that solves the problem and has running time
\begin{equation}\label{eq:FPT}
    \mathcal{O}(f(r)\cdot p(n))
\end{equation}
where $p$ is some polynomial and $n$ is the size of the input to the problem.
Many $\mathbb{NP}$-Complete problems are fixed-parameter tractable, which means that their time complexity is not necessarily super-polynomial in the input size but rather in the parameter $r$.
For $\mathbb{NP}$-Complete fixed-parameter tractable problems, the factor $f(r)$ must scale super-polynomially with $n$ in the worst case, unless $\mathbb{P}=\mathbb{NP}$.

In 1990, Courcelle proved that a large class of graph problems are fixed-parameter tractable in the tree-width of the graph~\cite{Courcelle1990}.
Courcelle's theorem states that any graph problem specified by a monadic second-order logic (MS) formula can be solved in linear time on graphs of bounded tree-width.
The tree-width is a notion which essentially describes how tree-like a graph is~\cite{Robertson1986}.
Many problems that are hard in general become tractable on trees, as for example the subgraph isomorphism problem\cite{Gupta1996}.
The same holds for graphs which are not too different from trees, i.e. have a low tree-width, which is exactly what Courcelle's theorem states.
Since the original theorem by Courcelle, there has also been many generalizations including the same statement but using rank-width.
Bounded rank-width captures a larger class of graph than tree-width, for example complete graphs have minimal rank-width.
Importantly here is that rank-width is invariant under local complementations and non-increasing under vertex-deletions~\cite{Oum2005}.
We give more details on rank-width in appendix~\ref{sec:background}\ref{app:rwd}.

MS logic is an extension of first-order logic, which allows for quantification over sets~\cite{Courcelle2011}.
Courcelle's theorem actually holds for a strictly more expressive logic called counting monadic second-order logic (CMS) where one can also express whether the size of a set is zero modulo $p$~\cite{Courcelle2011}.
A sublanguage of CMS is C$_2$MS where $p$ is restricted to be $2$ and one can therefore express whether the size of a set is even or odd.

Any graph problem specified by a C$_2$MS formula can be solved in cubic time on graphs of bounded rank-width, which is due to theorem 6.55 in~\cite{Courcelle2011}.
We state this formally in theorem~\ref{thm:C2MS}.
This is the result we make use of in this section to find an efficient algorithm for the vertex-minor problem.

It turns out that the vertex-minor problem is expressible in C$_2$MS, which we formally state in theorem~\ref{thm:VM_expression}.
By theorem~\ref{thm:C2MS} and theorem~\ref{thm:VM_expression} we see that the problem of deciding whether $\ket{G'}$ is a qubit-minor $\ket{G}$ is fixed-parameter tractable in the rank-width of $G$ and in the size of $G'$ in general, as we captured in theorem~\ref{thm:QM_FPT}.
The reason the qubit-minor problem is fixed-parameter tractable in both $\rwd(G)$ and $\abs{G'}$ is because the formula $\textproc{VM}_{G'}$ in equation~\eqref{eq:VM_expression} depends on $G'$.
Note that if conjecture~\ref{conj:qrFPT} is true this dependence of $\abs{G'}$ in the running time can be removed.
If $G'$ is restricted to be a certain type of graph or if we ask the question whether $G$ has a vertex-minor on the subset $U\subseteq V(G)$ with a given property instead, then the running time does not need to depend $G'$ or $U$ respectively, see theorems~\ref{thm:GHZ_FPT} and~\ref{thm:prop_VM}, even if conjecture~\ref{conj:qrFPT} is false.

Algorithm~\ref{alg:courcelle} below gives a high-level description of how to solve the vertex-minor problem efficiently on graphs of bounded rank-width.
The overall runtime of the algorithm is $\mathcal{O}(f(\abs{G'},r)\cdot\abs{G}^3)$ and is dominated by line 6.
In line 5, the C$_2$MS formula $\textproc{VM}'$ defined in equation~\eqref{eq:vmprime} can be constructed in time $\mathcal{O}(\abs{G'})$.
An assignment $\alpha_y$ as in line 6 can be found in time $\mathcal{O}(f(\abs{G'},r)\cdot\abs{G}^3)$ due to Courcelle's theorem for CMS selection problems~\cite{Courcelle2011} together with a proof similar to theorem~\ref{thm:C2MS}.
Finally a sequence of switchings $m$ taking the Eulerian vector $(\emptyset,V(G),\emptyset)$ to the $\alpha_y(X_e,Y_e,Z_e)$ can be done in time $\mathcal{O}(\abs{G})$ as we show in appendix~\ref{app:LCfromE}.

\begin{algorithm}[H]
    \caption{Algorithm that decides if $G'<G$. Runtime: $\mathcal{O}(f(\abs{G'},r)\cdot\abs{G}^3)$ \newline Input: ($G$,$G'$) where $\rwd(G)=r$. \newline Output: A sequence $m$ such that $\tau_m(G)[V(G')]=G'$ if $G'<G$.\newline \hbox{}\hspace{1.2cm}$\perp$ \hspace{5.85cm}if $G'\nless G$.}\label{alg:courcelle}
    \begin{algorithmic}[1]
        \Function{is\_vm}{$G$,$G'$}
            \If{$V(G')\nsubseteq V(G)$}
                \State \textbf{return} $\perp$
            \EndIf
            \State Construct the C$_2$MS formula $\textproc{VM}'_{G'}(\mathcal{X},X_e,Y_e,Z_e)$\Comment{See eq.~\ref{eq:vmprime}.}
            \State Find an assignment $\alpha_y$ such that $G\models\textproc{VM}'_{G'}(\mathcal{X}\mapsto V(G'),\alpha_y(X_e,Y_e,Z_e))$
            \If{There is no such $\alpha_y$}
                \State \textbf{return} $\perp$
            \Else
            \State Find a sequence of switchings $m$ taking $(\emptyset,V(G),\emptyset)$ to $\alpha_y(X_e,Y_e,Z_e)$\Comment{See sec.~\ref{sec:courcelle}\ref{subsec:sequence}}
                \State \textbf{return} $m$
            \EndIf
        \EndFunction
    \end{algorithmic}
\end{algorithm}

\subsection{Monadic second-order logic}\label{subsec:ms}
Monadic second-order logic is an sublanguage of second-order logic which in turn is an extension of first-order logic.
In MS one can quantify over sets\footnote{In second-order logic one can more generally quantify over predicates. MS is restricted to quantification over predicates with one argument (monadic), which is equivalent to quantification over sets.} compared to first-order logic which is restricted to only quantification over elements.
For a detailed reference on MS and its extensions, see the book by Courcelle and Engelfriet~\cite{Courcelle2011}.

MS formulas on graphs uses variables which are either vertex variables $x,y,\dots$ or set variables $X,Y,\dots$, which are sets of vertices. 
A MS formula is a finite string built up by the atomic formulas $x=y$, $x\in X$ and $\textproc{adj}(x,y)$\footnote{Expressing whether $(x,y)$ is an edge in the considered graph.} together with the following recursive rules~\cite{Langer2014}:
\begin{enumerate}
    \item If $\phi$ is a formula, then so is $\neg\phi$.
    \item If $\phi_1\dots\phi_l$ are formulas, then so are $(\phi_1\lor\dots\lor\phi_l)$ and $(\phi_1\land\dots\land\phi_l)$.
    \item If $\phi$ is a formula, then so are $\exists x:\phi$, $\forall x:\phi$, $\exists X:\phi$ and $\forall X:\phi$.
\end{enumerate}
The variables of a formula which are not part of a quantifier, as in rule (iii) above, are called free variables.
A formula with no free variables is called a sentence.
We write $\phi(X_1,\dots, X_l, x_1,\dots, x_m)$ for a formula with free variables $X_1\dots X_lx_1\dots x_m$.
To simplify formulas we will sometimes make use of the following abbreviations
\begin{align}
    (\phi\Rightarrow\psi)&\equiv(\neg\phi\lor\psi)\\
    (\phi\Leftrightarrow\psi)&\equiv((\phi\Rightarrow\psi)\land(\psi\Rightarrow\phi))\\
    (X\subseteq Y)&\equiv(\forall x:(x\in X\Rightarrow x\in Y)).
\end{align}

A MS formula is related to a graph $G$ by the atomic formula $\adj{x,y}$, which is true if and only if $(x,y)$ is an edge in $G$.
If a MS sentence $\phi$ is true on a graph $G$, we say that $G$ models $\phi$ and write this as $G\models\phi$.
For a formula with free variables, an assignment $\alpha$ is a mapping from the free variables to vertices and subsets of vertices of a graph $G$.
If a MS formula $\phi(X_1,\dots,X_l,x_1,\dots,x_m)$ is true on a graph $G$ with the assignment $\alpha$, we say that $\alpha$ satisfies $\phi$ on $G$ and write this as
\begin{equation}
    G\models\phi(\alpha(X_1),\dots,\alpha(X_l),\alpha(x_1),\dots,\alpha(x_m)).
\end{equation}
If $\alpha$ assigns $v$ to the free variable $x$, we write this as $x\mapsto v$.
Furthermore, if a formula has free variables $\mathcal{X}=X_1\dots X_lx_1\dots x_m$ then we sometimes write $\phi(\mathcal{X})$ and $\phi(\alpha(\mathcal{X}))$ for an assignment of these variables.
When considering a formula on a graph we implicitly assume that the quantifiers are over the vertex-set of the graph, i.e.
\begin{equation}
    G\models\forall x:\phi(x)\quad\text{iff}\quad \bigwedge_{v\in V(G)}G\models\phi(x\mapsto v).
\end{equation}
As an example the complete graph satisfies the following formula
\begin{equation}\label{eq:MSK}
    K_n\models\forall x,y:\big(\neg(x=y)\Rightarrow\adj{x,y}\big).
\end{equation}
There are many $\mathbb{NP}$-Complete problems that can be defined in MS, including for example 3-colorability~\cite{Ganian2010}.
In extensions of MS, one can also consider optimization problems such as minimum vertex cover and traveling sales person~\cite{Langer2014}.

\subsection{MS problems and complexity}\label{subsec:MSproblems}
Given a MS formula, there are multiple problems one can consider.
We will here be interested in model-checking, property-checking and selection problems, but will also include listing, counting and optimizing for completeness.
Below we define these different problems and further details can be found in~\cite{Courcelle2011}.
\begin{itemize}
    \item Model-checking: Given a sentence $\phi$ and a graph, decide if $G\models\phi$.
    \item Property-checking: Given a formula $\phi(\mathcal{X})$, a graph $G$ and an assignment $\alpha$, decide if $G\models\phi(\alpha(\mathcal{X}))$.
    \item Selection: Given a formula $\phi(\mathcal{X})$ and a graph $G$, find an assignment $\alpha$ such that $G\models\phi(\alpha(\mathcal{X}))$ or if there is no such assignment output $\perp$.
    \item Listing: Given a formula $\phi(\mathcal{X})$ and a graph $G$, find the set of assignments $\{\alpha\,:\,G\models\phi(\alpha(\mathcal{X}))\}$.
    \item Counting: Given a formula $\phi(\mathcal{X})$ and a graph $G$, find the size of the set of assignments $\{\alpha\,:\,G\models\phi(\alpha(\mathcal{X}))\}$.
    \item Optimizing: Given a formula $\phi(X)$ and a graph $G$, find the maximum cardinality of a set assigned to $X$, i.e. $\max(\{\abs{\alpha(X)}\,:\,G\models\phi(\alpha(X))\})$.
\end{itemize}
In the above definitions we have only considered the problem of finding assignments to all the free variables of a formula for simplicity.
One can also consider similar problems as above, where one is given a formula $\phi(\mathcal{X},\mathcal{Y})$, a graph $G$, an assignment $\alpha_y$ to the free variables $\mathcal{Y}$ and where the task is to find an assignment to the rest of the free variables, i.e. a $\alpha_x$ such that $G\models\phi(\alpha_x(\mathcal{X}),\alpha_y(\mathcal{Y}))$.
It turns out that all the above problems are fixed-parameter tractable\footnote{Note that the output of the listing problem is possibly super-polynomial in the size of the graph. Thus for the listing problem, fixed-parameter tractable means polynomial scaling in the size of the input plus the size of the input.} in the formula and the clique-width of the graph, as shown in~\cite{Courcelle2011}.
The same is therefore true for rank-width, since rank-width is bounded if and only if clique-width is bounded~\cite{Oum2006}.
\begin{theorem}\label{thm:C2MS}
    There exists an algorithm which checks whether an assignment $\alpha$ satisfies a C$_2$MS formula $\phi(\mathcal{X})$ on $G$, i.e. whether $G\models\phi(\alpha(\mathcal{X}))$ or not, and has a running time
    \begin{equation}\label{eq:scaling}
        \mathcal{O}(f(\abs{\phi},\rwd(G))\cdot \abs{G}^3).
    \end{equation}
    where $\rwd(G)$ is the rank-width of $G$ and $f$ is an computable function.
\end{theorem}
\begin{proof}
    Theorem 6.55 in~\cite{Courcelle2011} states that CMS model-checking problem can be solved in time
    \begin{equation}\label{eq:time_courcelle}
        \mathcal{O}(f(\phi,\cwd(G))\cdot \abs{G}^3),
    \end{equation}
    where $\cwd(G)$ is the clique-width of $G$.
    In section 6.4.1 of~\cite{Courcelle2011} it is also shown that the CMS property-checking problem can be reduced to the CMS model-checking problem and can therefore equivalently be solved in time as in equation~\eqref{eq:time_courcelle}.
    Furthermore, in definition 6.1 of~\cite{Courcelle2011} they show that the CMS model-checking problem can be solved in time
    \begin{equation}
        \mathcal{O}(f(\abs{\phi},\cwd(G))\cdot \abs{G}^3),
    \end{equation}
    since there are only finitely many sentences of size bounded by a given integer.
    The theorem then follows since the rank-width is bounded if and only if the clique-with is bounded, which is because
    \begin{equation}
        \rwd(G)\leq\cwd(G)\leq2^{\rwd(G)+1}-1,
    \end{equation}
    as shown by Oum in~\cite{Oum2006}.
\end{proof}

In many cases it is in fact the quantifier rank, see definition~\ref{def:qr} below, which dominates the runtime to solve MS problems.
For example in~\cite{Langer2011} it is shown by Langer, Rossmanith and Sikdar that the MS model-checking problem, i.e. if $G\models\phi$, can be solved in time $\mathcal{O}(f(\qr(\phi),\cwd(G))\cdot\abs{G}^3))$.
As discussed with Langer, it is probably possible to extend this statement to also the CMS model-checking problem and therefore CMS property-checking problem.
We therefore make the following conjecture,
\begin{conjecture}\label{conj:qrFPT}
    There exists an algorithm which checks whether an assignment $\alpha$ satisfies a C$_2$MS formula $\phi(\mathcal{X})$ on $G$, i.e. whether $G\models\phi(\alpha(\mathcal{X}))$ or not, and has a running time
    \begin{equation}\label{eq:scaling_qr}
        \mathcal{O}(f(\mathrm{qr}(\phi),\rwd(G))\cdot \abs{G}^3).
    \end{equation}
    where $\mathrm{qr}(\phi)$ is the quantifier rank of $\phi$, $\rwd(G)$ is the rank-width of $G$ and $f$ is a computable function.
\end{conjecture}
\begin{definition}\label{def:qr}
    The quantifier rank $\qr(\phi)$ of a formula $\phi$ is the maximum number of nested quantifiers as defined in~\cite{Langer2014} and can be found by the following recursive relations.
    \begin{align}
        \qr(\phi)         &=0\text{, if }\phi\text{ is atomic}       &\qr(\exists x:\phi(x))&=\qr(\phi)+1\\
        \qr(\neg\phi)     &=\qr(\phi)                                &\qr(\exists X:\phi(X))&=\qr(\phi)+1\\
        \qr(\phi\lor\psi) &=\max(\{\qr(\phi),\qr(\psi)\})            &\qr(\forall x:\phi(x))&=\qr(\phi)+1\\
        \qr(\phi\land\psi)&=\max(\{\qr(\phi),\qr(\psi)\})            &\qr(\forall X:\phi(X))&=\qr(\phi)+1
    \end{align}
\end{definition}

\subsection{Vertex-minor as C\texorpdfstring{$_2$}{2}MS formula}\label{subsec:vm_formula}
In~\cite{Courcelle2007} Courcelle and Oum show how one can express whether a graph $G'$ is a vertex-minor of $G$ in counting monadic second-order logic C$_2$MS.
We restate this here and also provide the explicit C$_2$MS formula and its quantifier rank.
In appendix~\ref{app:formula} we explicitly provide all subformulas we make use of here, which can also be found in~\cite{Courcelle2007}.
These formulas and the statement of this section build heavily on the concept of an \emph{isotropic system} which was introduced by Bouchet in~\cite{Bouchet1987}.
We will not go into the details of isotropic systems here; more details can be found in~\cite{Bouchet1987}.
The reason for the relation to isotropic systems is that an isotropic system describes an equivalence class of graphs under local complementation.
Importantly, an isotropic system\footnote{There are actually multiple isotropic system $S(G,A,B)$ related to a graph, depending on the choice of supplementary vectors $A$ and $B$, see~\cite{Bouchet1987}. Here we consider a canonical isotropic system $S(G)$ and chose the supplementary vectors to be $A=(\omega,\dots,\omega)$ and $B=(1,\dots,1)$, where $\omega$ is a primitive element of $\mathbb{F}_4$} $S(G)$ given by a graph $G$ has a number of \emph{Eulerian vectors} and each of these Eulerian
vectors describes a LC-equivalent graph to $G$.
Furthermore, any LC-equivalent graph to $G$ is described by some Eulerian vector of $S(G)$.
We will here describe a Eulerian vector by three pairwise disjoint subsets of the vertices of $G$ whose union is $V(G)$ and write this as a tuple $(X_e,Y_e,Z_e)$.
The set of Eulerian vectors of $S(G)$ will be denoted as $\mathcal{E}(S(G))$.
The formula $\textproc{Eul}(X_e,Y_e,Z_e)$ in equation~\eqref{eq:MS_Eul} describes whether $(X_e,Y_e,Z_e)$ is a Eulerian vector of $S(G)$, i.e.
\begin{equation}
    G\models\textproc{Eul}(X_e\mapsto U_e,Y_e\mapsto V_e,Z_e\mapsto W_e)\quad\text{iff}\quad (U_e,V_e,W_e)\in \mathcal{E}(S(G)).
\end{equation}

As mentioned above, the set of graphs described by the Eulerian vectors of $S(G)$ are exactly the LC-equivalent graphs to $G$.
Bouchet used this to develop an efficient algorithm to test LC-equivalence between graphs in~\cite{Bouchet1991}, which has been used to efficiently test single-qubit Clifford equivalence of graph states in~\cite{VandenNest2004a}.
Let's denote the LC-equivalent graph to $G$ described by the Eulerian vector $(X_e,Y_e,Z_e)$ as $\mathcal{G}(X_e,Y_e,Z_e)$.
We will now find a formula that captures whether a graph $G'$ is a vertex-minor of $G$.
From the above and equation~\eqref{eq:LC_VM} we have that $G'<G$ if and only if there exists a Eulerian vector $(X_e,Y_e,Z_e)$ such that 
\begin{equation}
    \mathcal{G}(X_e,Y_e,Z_e)[V(G')]=G'.
\end{equation}
How do we express this as a C$_2$MS formula?
The formula $\textproc{Adj}(u,v,X_e,Y_e,Z_e)$ in equation~\eqref{eq:MS_Adj} describes whether the edge $(u,v)$ is an edge of the graph $\mathcal{G}(X_e,Y_e,Z_e)$, i.e.
\begin{equation}\label{eq:adj_rel}
    G\models\textproc{Adj}(x\mapsto u,y\mapsto v,X_e\mapsto U_e,Y_e\mapsto V_e,Z_e\mapsto W_e)\quad\text{iff}\quad (u,v)\in E(\mathcal{G}(U_e,V_e,W_e)).
\end{equation}
Using equation~\eqref{eq:adj_rel} we can express whether $G'$ is a vertex-minor of $G$, as described in the following theorem.
\begin{theorem}\label{thm:VM_expression}
    For any $G'$ with vertex-set $V(G')=\{x_1,\dots,x_k\}$, there exists a C$_2$MS formula $\textproc{VM}_{G'}(x_1,\dots,x_k)$ such that
    \begin{equation}\label{eq:VM_expression}
        G\models\textproc{VM}_{G'}(\alpha(x_1),\dots,\alpha(x_k))\quad\text{iff}\quad\alpha(G')<G,
    \end{equation}
    where $x_i$ are the free variables of the formula, $\alpha$ is a bijection from $V(G')$ to a subset of $V(G)$ of size $k$.
    In equation~\eqref{eq:VM_expression}, $\alpha$ functions both as an assignment of the free variables of $\textproc{VM}$ and as a relabeling of the vertices in $G'$ by $\alpha(G)$.
    This dual-purpose of $\alpha$ is valid since we identify the free variables of $\textproc{VM}$ with the vertices of $G'$.
    To be precise, by $\alpha(G)$ we mean the graph 
    \begin{equation}
        \alpha(G)=(\{\alpha(x)\;:\;x\in V(G')\},\{(\alpha(x),\alpha(y))\;:\;(x,y)\in E(G')\}).
    \end{equation}
    Furthermore the length and quantifier rank of $\textproc{VM}$ has the following scaling
    \begin{equation}
        \abs{\textproc{VM}_{G'}}=\mathcal{O}(\abs{G'}^2),\quad\qr(\textproc{VM}_{G'})=10=\mathcal{O}(1).
    \end{equation}
\end{theorem}
\begin{proof}
    We prove this by explicitly providing the C$_2$MS formula as follows
    \begin{align}
        \textproc{VM}_{G'}(x_1,\dots,x_k)=\exists X_e,Y_e,Z_e:\Bigg(\textproc{Eul}(X_e,Y_e,Z_e)&\land\bigwedge_{(x,y)\in E(G')}\textproc{Adj}(x,y,X_e,Y_e,Z_e) \nonumber\\
                                                                                               &\land\bigwedge_{(x,y)\notin E(G')}\neg\textproc{Adj}(x,y,X_e,Y_e,Z_e) \Bigg)
    \end{align}
    It is then clear that equation~\eqref{eq:VM_expression} is true, since if $G\models\textproc{VM}(\alpha(x_1),\dots,\alpha(x_k))$ then we know that there exist an LC-equivalent graph to $G$ which induced subgraph on $V(G')$ has the edge-set $\{(\alpha(x),\alpha(y)):(x,y)\in E(G')\}$.
    This is precisely $\alpha(G')$ and therefore $\alpha(G')<G$.
    Furthermore, if $\alpha(G')<G$ then we know that there exist a Eulerian vector $(U_e,V_e,W_e)$ such that
    \begin{equation}
        \bigwedge_{(x,y)\in E(G')}\textproc{Adj}(\alpha(x),\alpha(y),U_e,V_e,W_e)\land\bigwedge_{(x,y)\notin E(G')}\neg\textproc{Adj}(\alpha(x),\alpha(y),U_e,V_e,W_e)
    \end{equation}
    is true.

    Next we show how the length and quantifier rank of $\textproc{VM}$ scale with $G'$.
    Firstly the length clearly scales as
    \begin{equation}
        \abs{\textproc{VM}}=\mathcal{O}(\abs{E(G')})=\mathcal{O}(\abs{G'}^2).
    \end{equation}
    The quantifier ranks of the subformulas used here are given in appendix~\ref{app:formula} and in particular we have that $\qr(\textproc{Eul})=7$ and $\qr(\textproc{Adj})=7$.
    Thus, we have that the quantifier rank of $\textproc{VM}$ is
    \begin{equation}
        \qr(\textproc{VM})=3+\max(\{\qr(\textproc{Eul}),\qr(\textproc{Adj})\})=3+\max(\{7,7\})=10=\mathcal{O}(1).
    \end{equation}
\end{proof}
It is easy to see that theorem~\ref{thm:QM_FPT} is a direct consequence of theorem~\ref{thm:QMVM}, theorem~\ref{thm:C2MS} and theorem~\ref{thm:VM_expression}.

As mentioned earlier, it is possible to specify in C$_2$MS  whether a graph has a vertex-minor on a subset of the vertices with a given property.
We capture this in the following theorem.
\begin{theorem}\label{thm:prop_VM}
    Given a C$_2$MS sentence $\phi$ specifying some graph property $P$, then there exists a C$_2$MS formula $\textproc{Prop\_VM}_\phi(X)$ capturing whether a graph has a vertex-minor on $X$ which satisfies $P$, i.e.
    \begin{equation}
        G\models\textproc{Prop\_VM}_\phi(X\mapsto U)\quad\text{iff}\quad\exists G':(V(G')=U)\land(G'<G)\land(G'\models\phi)
    \end{equation}
\end{theorem}
\begin{proof}
    Let $\phi'$ be the formula made from $\phi$ by replacing all instances of the predicate $\adj{x,y}$ by the formula $\textproc{Adj}(x,y,X_e,Y_e,Z_e)$.
    If $(U_e,V_e,W_e)$ is a Eulerian vector of $S(G)$, then we have that
    \begin{equation}
        G\models\phi'(X_e\mapsto U_e,Y_e\mapsto V_e,Z_e\mapsto W_e)\quad\text{iff}\quad \mathcal{G}(U_e,V_e,W_e)\models\phi.
    \end{equation}
    The expression on the right of the above equation states that the LC-equivalent graph $\mathcal{G}(U_e,V_e,W_e)$ models $\phi$, but what we want is that the induced subgraph $\mathcal{G}(U_e,V_e,W_e)[X]$ models $\phi$.
    This can be done by restricting all the quantifiers in $\phi'$ to the set $X$.
    Thus, let $\phi''$ be the formula made from $\phi'$ by making the following changes to all the quantifiers
    \begin{align}
        \forall y:\psi(y)\quad&\rightarrow\quad\forall y:\big(y\in X\Rightarrow\psi(y)\big) \label{eq:replq1}\\
        \forall Y:\psi(Y)\quad&\rightarrow\quad\forall Y:\big(Y\subseteq X\Rightarrow\psi(Y)\big) \label{eq:replq2}\\
        \exists y:\psi(y)\quad&\rightarrow\quad\exists y:\big(y\in X\land\psi(y)\big) \label{eq:replq3}\\
        \exists Y:\psi(Y)\quad&\rightarrow\quad\exists Y:\big(Y\subseteq X\land\psi(Y)\big) \label{eq:replq4}.
    \end{align}
    We then see that if $(U_e,V_e,W_e)$ is an Eulerian vector of $S(G)$, then
    \begin{equation}
        G\models\phi''(X\mapsto U,X_e\mapsto U_e,Y_e\mapsto V_e,Z_e\mapsto W_e)\quad\text{iff}\quad \mathcal{G}(U_e,V_e,W_e)[U]\models\phi.
    \end{equation}
    The formula $\textproc{Prop\_VM}$ can then be built by checking if there exists a Eulerian vector such that $G$ models $\phi''$, i.e.
    \begin{equation}\label{eq:prop_VM}
        \textproc{Prop\_VM}_\phi(X)=\exists X_e,Y_e,Z_e:\Big[\textproc{Eul}(X_e,Y_e,Z_e)\land\phi''(X,X_e,Y_e,Z_e)\Big].
    \end{equation}
\end{proof}
Since it is possible to specify whether a graph is a complete graph in C$_2$MS, see equation~\eqref{eq:MSK}, we see that theorem~\ref{thm:GHZ_FPT} directly follows from theorem~\ref{thm:QMVM}, theorem~\ref{thm:C2MS} and theorem~\ref{thm:prop_VM}.
Let's use the method in the proof of theorem~\ref{thm:prop_VM} to explicitly find the formula expressing whether a graph has the complete graph as a vertex-minor on the subset $X$.
Let $\textproc{Complete}$ be the formula in equation~\eqref{eq:MSK} which is modeled by $G$ if $G$ is a complete graph.
We first find $\textproc{Complete}'$ by replacing the predicate $\adj{x,y}$
\begin{equation}
    \textproc{Complete}'(X_e,Y_e,Z_e)=\forall x,y:\big(\neg(x=y)\Rightarrow\textproc{Adj}(x,y,X_e,Y_e,Z_e)\big).
\end{equation}
Then replacing the quantifiers as in equations~\eqref{eq:replq1}-\eqref{eq:replq4} we get the formula
\begin{align}
    \textproc{Complete}''(X,X_e,Y_e,Z_e)&=\forall x:\Big[x\in X\Rightarrow\Big(\forall y:y\in X\Rightarrow\big(\neg(x=y)\Rightarrow\textproc{Adj}(x,y,X_e,Y_e,Z_e)\big)\Big)\Big] \nonumber\\
                                        &=\forall x,y:\Big[\big((x\in X)\land(y\in X)\land\neg(x=y)\big)\Rightarrow\textproc{Adj}(x,y,X_e,Y_e,Z_e)\Big]
\end{align}
Finally by using the equation~\eqref{eq:prop_VM} we arrive at the formula
\begin{equation}
    \textproc{Complete\_VM}(X)=\exists X_e,Y_e,Z_e:\Big[\textproc{Eul}(X_e,Y_e,Z_e)\land\textproc{Complete}''(X,X_e,Y_e,Z_e)\Big]
\end{equation}
which has the following property
\begin{equation}
    G\models\textproc{Complete\_VM}(X\mapsto U)\quad\text{iff}\quad K_U<G
\end{equation}
where $K_U$ is the complete graph with vertex-set $U$.

\subsection{Finding the sequence of operations}\label{subsec:sequence}
In this section so far, we have looked at the problem of deciding whether $G'$ is a vertex-minor of $G$, but if this is true then how does one find the sequence of operations that takes $G$ to $G'$?
Similarly, if $\ket{G'}$ is a qubit-minor of $\ket{G}$, what sequence of operations takes $\ket{G}$ to $\ket{G'}$?
We will here describe two ways to find the sequence of operations.

\noindent\underline{Method 1:}
The first way is slightly simpler but increases the runtime from the decision problem by a factor of $\abs{G}$.
The idea is to use an algorithm that solves the decision problem of whether $G'$ is a vertex-minor of $G$ to iteratively find the sequence of operations.
Let's therefore assume that we know that $G'<G$.
Furthermore, let $v$ be a vertex in $V(G)\setminus V(G')$.
From theorem~\ref{thm:multi_vertex-minor} we know that $G'$ is a vertex-minor of at least one of the three graphs $G\setminus v$, $\tau_v(G)\setminus v$ or $T_v(G)\setminus v$.
By using an algorithm for the decision problem we can decide which of these three graphs has $G'$ as a vertex-minor.
Let's denote one of these graphs by $G_1$ and the operation that takes $G$ to $G_1$ as $P_1$, i.e. $G_1=P_1(G)$.
That is, $P_1$ is either $({\_\,})\setminus v$, $\tau_v({\_\,})\setminus v$ or $T_v({\_\,})\setminus v$, such that $G'<G_1<G$.
Now perform the same step again to find an operation $P_2$ taking $G_1$ to a graph $G_2$ which has $G'$ as a vertex-minor.
Perform the step $n-k$ times, where $n=\abs{G}$ and $k=\abs{G'}$.
It is then clear that the sequence $P=P_{n-k}\circ\dots\circ P_1$ takes $G$ to a LC-equivalent graph of $G'$.
From $P$ it is easy to find the induced sequence of local complementations $m$ such that $\tau_m(G)[V(G')]\sim_\mathrm{LC}G'$.
Finally we can use the algorithm in~\cite{Bouchet1991} to find a sequence of local complementations $m'$ such that $\tau_{m'}\circ\tau_m(G)[V(G')]=G'$.
Assume that $\abs{G'}$ and $\rwd(G)$ are bounded\footnote{If conjecture~\ref{conj:qrFPT} is true, then $\abs{G'}$ does not need to be bounded.}.
Then according to theorem~\ref{thm:QM_FPT} the decision problem can be solved in time $\mathcal{O}(n^3)$.
To find $P$ we need to run the algorithm for the decision problem $\mathcal{O}(n-k)$ times and compute $P_i(G_{i-1})$ the same number of times.
Computing $P_i(G_{i-1})$ can be done in time $\mathcal{O}(n^2)$, as described in section~\ref{sec:background}\ref{sec:brute}.
Lastly finding the sequence $m'$ can be done in time $\mathcal{O}(k^4)$.
Thus, the total runtime is
\begin{equation}
    \mathcal{O}((n^3+n^2)(n-k))+\mathcal{O}(k^4)=\mathcal{O}(n^4)
\end{equation}

\noindent\underline{Method 2:}
Another way to find the sequence of operations that takes $G$ to $G'$, given that $G'<G$, is to formulate the problem as a C$_2$MS selection problem as described in section~\ref{sec:courcelle}\ref{subsec:MSproblems}.
Recall that an algorithm that solves the C$_2$MS selection problem takes as input a formula $\phi$ with free variables $\mathcal{X}$, $\mathcal{Y}$, an assignment $\alpha_x$ to the free variables $\mathcal{X}$ and a graph $G$ and returns an assignment $\alpha_y$ such that $G\models\phi(\alpha_x(\mathcal{X}),\alpha_y(\mathcal{Y}))$ or returns $\perp$ if no such assignment exists.
We will now formulate a selection problem by making $X_e$, $Y_e$ and $Z_e$ free variables in equation~\eqref{eq:VM_expression} instead of quantifier variables.
Therefore, let $H$ be an isomorphic graph to $G'$ with vertex-set $V(H)=\{x_1,\dots,x_k\}$ and define the C$_2$MS formula
\begin{equation}\label{eq:vmprime}
    \textproc{VM}'_H(x_1,\dots,x_k,X_e,Y_e,Z_e)=\textproc{Eul}(X_e,Y_e,Z_e)\land\bigwedge_{(x,y)\in E(H)}\textproc{Adj}(x,y,X_e,Y_e,Z_e)
\end{equation}
Let $\alpha_x$ be a bijection from $V(H)$ to $V(G')$ such that $\alpha_x(H)=G'$.
We then have the following property of the formula $\textproc{VM}'$
\begin{equation}
    \exists \alpha_e:G\models\textproc{VM}'_H(\alpha_x(x_1,\dots,x_k),\alpha_e(X_e,Y_e,Z_e))\quad\text{iff}\quad G'<G.
\end{equation}
From theorem 6.55 and similar arguments as in the proof of theorem~\ref{thm:C2MS} we know that we can solve the selection problem in time $\mathcal{O}(n^3)$ if $k$ and $\rwd(G)$ are bounded, where $n=\abs{G}$ and $k=\abs{G'}$.
Thus, given an assignment to the free variables $(X_e,Y_e,Z_e)$, i.e. given an Eulerian vector $(U_e,V_e,W_e)$ such that $\mathcal{G}(U_e,V_e,W_e)[V(G')]=G'$, the question is then how to find a sequence of local complementations that takes $G$ to $\mathcal{G}(U_e,V_e,W_e)$.
In appendix~\ref{app:LCfromE} we show how to do this in time $\mathcal{O}(n)$, which shows that the total runtime is
\begin{equation}
    \mathcal{O}(n^3)+\mathcal{O}(n)=\mathcal{O}(n^3).
\end{equation}

\section{Discussion}\label{sec:conclusion}
The problem of deciding whether a stabilizer state $\ket{S_t}$ can be obtained from another $\ket{S_s}$ by single-qubit Clifford operations, single-qubit Pauli measurement and classical communication is equivalent to deciding if some graph $G'$ is a vertex-minor of another graph $G$.
We showed here that the vertex-minor problem can be solved in cubic time in the size of $G$ on instances where $G$ has bounded rank-width and $G'$ has bounded size, by using the theory of monadic second-order logic and a version of Courcelle's theorem.
Furthermore, if conjecture~\ref{conj:qrFPT} is true then the vertex-minor problem can be solved in cubic time in the size of $G$ on the strictly larger class of instances where $G$ has bounded rank-width and $G'$ is arbitrary.
A direct implementation of Courcelle's theorem is however not practical, due to a huge constant factor in the runtime of the algorithm.
Finding more tailored algorithms for the vertex-minor problem on graphs of bounded rank-width is therefore of value.
In~\cite{npcomplete} we provide an efficient algorithm for graphs of rank-width one, which does not have a huge constant factor in the runtime.

Given some graph property $P$ expressible in C$_2$MS one can decide if a graph $G$ has a vertex-minor on a subset $U\subseteq V(G)$ that satisfies $P$, in cubic time in the size of $G$ for graphs with bounded rank-width, as we show in section~\ref{sec:courcelle}\ref{subsec:vm_formula}.
The graph property $P$ can be for example that the graph is a complete graph or that it is $k$-colourable.
Testing for for example 2-colourable qubit-minors could be interesting in the context of purification since it has been shown that the purification schemes in~\cite{Aschauer2005} purify all 2-colourable graph states.

In section~\ref{sec:courcelle}\ref{subsec:sequence} we also showed how to find the sequence of operations that take $\ket{G}$ to its qubit-minor $\ket{G'}$.
Finally in section~\ref{sec:constant} we showed how these operations can be applied in constant time in the size of the graph states and how this can be done without destroying all the rest of the entanglement in the source state.
An open question is how to find an optimal sequence of operations that destroys a minimum amount of entanglement in the rest of the state.

\enlargethispage{20pt}


\dataccess{The implementation of concepts and algorithms, written in SAGE and MONA, can be freely accessed from the git-repository at~\cite{git}.}

\aucontribute{Axel Dahlberg developed the theory and proofs, drafted most of the manuscript and wrote the code at~\cite{git}. 
Stephanie Wehner defined and supervised the project. Both authors read and approved the manuscript.}


\funding{Axel Dahlberg and Stephanie Wehner were supported by STW Netherlands, and NWO VIDI grant, and an ERC Starting grant.}

\ack{We thank Alexander Langer, Robert Ganian, Kenneth Goodenough, Jonas Helsen, Tim Coopmans, J\'{e}ri\'{e}my Ribiero and Ben Criger for interesting discussions and feedback.}


\appendix
\section{Corrections from sequence of Pauli \texorpdfstring{$Z$}{Z} measurement}\label{app:correct}
By performing a measurement in the standard basis of a qubit $v$ which is part of a graph state $\ket{G}$, one can effectively disconnect qubit $v$ from the rest of the state and produce the state $\ket{0}_v\otimes\ket{G\setminus v}$.
Depending on the measurement outcome, certain single-qubit Clifford operations need to be performed to map the post-measurement state to $\ket{0}_v\otimes\ket{G\setminus v}$, as described in section~\ref{sec:background}\ref{subsec:meas}.
One can therefore effectively cut out a graph state on a subset of the qubits $V'$, i.e. producing the state $\ket{G[V']}$, by measuring the qubits in $V(G)\setminus V'$ in the standard basis and performing certain single-qubit Clifford operations.
Here we show what corrections need to be applied to the qubits $V'$, such that the post-measurement state is mapped to the state $\ket{G[V']}$.

Let's assume that $\ket{G}$ is a graph state and we wish to transform this to $\ket{G[V']}$, by measuring the qubits $U=V(G)\setminus V'$ in the standard basis.
Let's denote the qubits in $U$ as $\{v_1,v_2,\dots,v_{n-k}\}$ and the measurement outcome of qubit $v_i$ by\footnote{We identify $0$ and $1$ with the measurement outcomes $+1$ and $-1$, respectively.} $x_i\in\{0,1\}$.
Furthermore, let's denote the projectors in the $Z$ basis as $P_v^{(0)}=P_v^{(Z,+)}$ and $P_v^{(1)}=P_v^{(Z,-)}$.
The post-measurement state is then given by
\begin{equation}
    \ket{\psi^{n-k}_\mathrm{post}}=2^{n-k}P_{v_{n-k}}^{(x_{n-k})}\cdot\dots\cdot P_{v_2}^{(x_2)}P_{v_1}^{(x_1)}\ket{G}.
\end{equation}
By acting with the projectors on $\ket{G}$ we find by induction on $n-k$ that the post-measurement state can be evaluated to
\begin{equation}\label{eq:poststate}
    |\psi^{n-k}_\mathrm{post}\rangle=\left(\bigotimes_{i=1}^{n-k}Z^{\sum_{j=1}^{i-1}x_j\mathrm{adj}(v_i,v_j)}|x_i\rangle_{v_i}\right)\otimes\left(\left(\prod_{i=1}^{n-k}(Z[N_{v_i}\cap V'])^{x_i}\right)|G[V']\rangle\right),
\end{equation}
where $\mathrm{adj}(u,v)$ is $1$ if $(u,v)$ is an edge in $G$ and zero otherwise and where $Z[X]$ is $\prod_{x\in X}Z_x$.
One can see this by checking that indeed
\begin{equation}
    2P_{v_{n-k}}^{(x_{n-k})}\ket{\psi^{n-k-1}_\mathrm{post}}=\ket{\psi^{n-k}_\mathrm{post}}.
\end{equation}
by using equations~\eqref{eq:meas_Z} and~\eqref{eq:corr}.
The operations on the qubits in $V(G)\setminus V'$ in the left part of equation~\ref{eq:poststate} will only give a global phase as follows
\begin{align}
    \left(\bigotimes_{i=1}^{n-k}Z^{\sum_{j=1}^{i-1}x_j\mathrm{adj}(v_i,v_j)}|x_i\rangle_{v_i}\right)&=\left(\bigotimes_{i=1}^{n-k}(-1)^{x_i\sum_{j=1}^{i-1}x_j\mathrm{adj}(v_i,v_j)}|x_i\rangle_{v_i}\right) \nonumber\\
                                                                                            &=(-1)^{\sum_{i=1}^{n-k}\sum_{j=1}^{i-1}x_ix_j\mathrm{adj}(v_i,v_j)}\bigotimes_{i=1}^{n-k}|x_i\rangle_{v_i}.
\end{align}
The exponent in the global phase $\sum_{i=1}^{n-k}\sum_{j=1}^{i-1}x_ix_j\mathrm{adj}(v_i,v_j)$ is in fact the number of edges in the induced graph $G_x=G[\{v_i:x_i=1\}]$.
Let's now consider the correction operators in the right part of equation~\ref{eq:poststate}.
A qubit $v\in V'$ will have a $Z$ contribution from the $i$th factor if $v\in N_{v_i}$ and $x_i=1$.
Thus, we see that the total contribution on qubit $v$ is given by
\begin{equation}
    Z^{y_v}\quad\text{where}\quad y_v=\sum_{i\in\{i\,:\,v_i\in N_v\setminus V'\}}x_i
\end{equation}
Finally we find that the post-measurement state from equation~\ref{eq:poststate} is given by
\begin{equation}
    |\psi^{n-k}_\mathrm{post}\rangle=(-1)^{\abs{E(G_x)}}\left(\bigotimes_{i=1}^{n-k}|x_i\rangle_{v_i}\right)\otimes\left(\left(\prod_{v\in V'}Z_v^{y_v}\right)|G[V']\rangle\right).
\end{equation}

\section{Vertex-minor formula}\label{app:formula}
Here we provide the C$_2$MS formulas\footnote{
    Note that there seems to be a typo in~\cite{Courcelle2007} since they use the formula $V=X_e\cup Y_e\cup Z_e$ to express that the vector $(X_e,Y_e,Z_e)$ is complete, i.e. that each element of the vector is non-zero.
    This is however not true, consider for example the sets $X_e=Y_e=Z_e=V$ which corresponds to the zero-vector since $1+\omega+\omega^2=0$ and is therefore not complete.
    Their formula for whether $(X_e,Y_e,Z_e)$ is a Eulerian vector is on the other hand still correct since the second part of the formula can never be true for a non-complete vector $(X_e,Y_e,Z_e)$ for which $V=X_e\cup Y_e\cup Z_e$.
}
which we make use of in section~\ref{sec:courcelle}\ref{subsec:vm_formula}.
We state what the formula expresses and its quantifier rank in table~\ref{tab:formulas}.

\begin{table}[!h]
    \caption{The C$_2$MS formulas used in section~\ref{sec:courcelle}\ref{subsec:vm_formula}, what they express and their quantifier rank. $b_v$ is the unique vector in $S(G)$ with respect to $(X_e,Y_e,Z_e)$ as defined in~\cite{Courcelle2007}.}
    \label{tab:formulas}
    \begin{tabular}{|c|l|c|}
        \hline
        \textbf{Formula} & \textbf{True if and only if} & \textbf{qr}\\
        \hline\hline
        $\textproc{Disjoint}(X,Y,Z)$ & $X$,$Y$ and $Z$ are pairwise disjoint & 1\\
        \hline
        $\textproc{Part}(X,Y,Z)$ & $(X,Y,Z)$ is a tripartition & 1\\
        \hline
        $\textproc{EvenInter}(Q,v)$ & $\abs{N_v\cap Q}=0\pmod{2}$ & 2\\
        \hline
        $\textproc{Member}(X,Y,Z)$ & $(X,Y,Z)$ is a vector of $S(G)$ & 4\\
        \hline
        $\textproc{Eul}(X_e,Y_e,Z_e)$ & $(X_e,Y_e,Z_e)$ is a Eulerian vector of $S(G)$ & 7\\
        \hline
        $\textproc{Base}(X,Y,Z,X_e,Y_e,Z_e)$ & $(X,Y,Z)$ is $b_v$ wrt. $(X_e,Y_e,Z_e)$ in $S(G)$ & 4\\
        \hline
        $\textproc{Adj}(u,v,X_e,Y_e,Z_e)$ & $(u,v)$ is an edge of $\mathcal{G}(X_e,Y_e,Z_e)$ & 7\\
        \hline
    \end{tabular}
    \vspace*{-4pt}
\end{table}

\begin{equation}
    \textproc{Disjoint}(X,Y,Z)=\forall x:\big(\neg(x\in X\land x\in Y)\land\neg(x\in X\land x\in Z)\land\neg(x\in Y\land x\in Z)\big)
\end{equation}
\begin{equation}
    \textproc{Part}(X,Y,Z)=\Big(\underbrace{\forall x:(x\in X\lor x\in Y\lor x\in Z)}_{\text{"}V=X\cup Y\cup Z\text{"}}\Big)\land\textproc{Disjoint}(X,Y,Z)
\end{equation}
\begin{equation}
    \textproc{EvenInter}(Q,v)=\forall R:\big(\underbrace{\forall u: u\in R\Leftrightarrow (\adj{u,v}\land u\in Q)}_{\text{"}R=N_v\cap Q\text{"}}\big)\Rightarrow\textproc{Even}(R)
\end{equation}
\begin{align}
    \textproc{Member}(X,Y,Z)=\textproc{Disjoint}(X,Y,Z)\land\Bigg[\exists Q:\Bigg(\forall v:\Big(\big(v\in X&\Leftarrow (v\notin Q\land\neg\textproc{EvenInter}(Q,v))\big)\land\nonumber\\
                                                                                             \big(v\in Y&\Leftarrow (v\in Q\land\textproc{EvenInter}(Q,v))\big)\land\nonumber\\
                                                                                             \big(v\in Z&\Leftarrow (v\in Q\land\neg\textproc{EvenInter}(Q,v))\big)\land\nonumber\\
                                                                                 \big(\underbrace{\neg(v\in X\lor v\in Y\lor v\in Z)}_{\text{"}v\in V(G)\setminus (X\cup Y\cup Z)\text{"}}&\Leftarrow (v\notin Q\land\textproc{EvenInter}(Q,v))\big)\Big)\Bigg)\Bigg]
\end{align}
\begin{align}
    \textproc{Eul}(X_e,Y_e,Z_e)=&(\textproc{Part}(X_e,Y_e,Z_e))\land\nonumber\\
                                                   &\Big[\forall X,Y,Z:\big(X\subseteq X_e\land Y\subseteq Y_e\land Z\subseteq Z_e\land\textproc{Member}(X,Y,Z)\big)\nonumber\\
                                &\Rightarrow(\underbrace{\forall v:\neg(v\in X\lor v\in Y\lor v\in Z)}_{\text{"}X=Y=Z=\emptyset\text{"}})\Big]\label{eq:MS_Eul}
\end{align}
\begin{align}
    \textproc{Base}(X,Y,Z,X_e,Y_e,Z_e,v)=&\textproc{Member}(X,Y,Z)\land (\overbrace{v\in X\lor v\in Y\lor V\in Z}^{\text{"}v\in X\cup Y\cup Z\text{"}})\land\nonumber\\
                                         &\Big[\forall u:\neg(v=u)\Rightarrow\nonumber\\
                                         &\Big((u\in X\Rightarrow u\in X_e)\land(u\in Y\Rightarrow u\in Y_e)\land(u\in Z\Rightarrow u\in Z_e)\Big)\Big]
\end{align}
\begin{align}
    \textproc{Adj}(u,v,X_e,Y_e,Z_e)=&\neg(u=v)\land\nonumber\\&\Big[\exists X,Y,Z:\big(\textproc{Base}(X,Y,Z,X_e,Y_e,Z_e,v)\land (\underbrace{u\in X\lor u\in Y\lor u\in Z}_{\text{"}u\in X\cup Y\cup Z\text{"}})\big)\Big]\label{eq:MS_Adj}
\end{align}

\section{Local complementations from Eulerian vector}\label{app:LCfromE}
Let's assume $G$ is a graph, $S(G)$ its canonical isotropic system and $(U_e,V_e,W_e)$ an Eulerian vector describing the graph $\mathcal{G}(U_e,V_e,W_e)$.
We here consider the question of how to find a sequence of local complementations $m$, such that $\tau_m(G)=\mathcal{G}(U_e,V_e,W_e)$.
In this section we will represent Eulerian vectors as vectors in $\mathbb{F}_4^n$ instead of tripartitions of $V(G)$, where $n=\abs{G}$.
A tripartition $(U_e,V_e,W_e)$ induces the vector
\begin{equation}\label{eq:eul_vec}
A(v)=\begin{cases}1 & \quad\text{if }v\in U_e\\\omega & \quad\text{if }v\in V_e\\\omega^2 & \quad\text{if }v\in W_e\end{cases}
\end{equation}
where $A(v)$ is element $v$ of the vector $A\in\mathbb{F}_4^n$ and $\omega$ is a primitive element of $\mathbb{F}_4$.
In~\cite{Bouchet1993} it is shown that for any Eulerian vector $A$ of an isotropic system there exists exactly one other Eulerian vector $A'$ which differ from $A$ in only the element $v$.
This other Eulerian vector $A'$ is denoted $A*v$ and is called a switching of $A$.
Furthermore, the switching induces a local complementation on the graphs the Eulerian vectors describe.
More precisely, if the Eulerian vector $A$ describe the graph $G$, then the Eulerian vector $A*v$ describe the graph $\tau_v(G)$.
Thus, if we find a sequence of switchings taking the Eulerian vector describing $G$ to the Eulerian vector describing $\mathcal{G}(U_e,V_e,W_e)$, we have also directly found a sequence of local complementations taking $G$ to $\mathcal{G}(U_e,V_e,W_e)$.
The Eulerian vector of $S(G)$ describing $G$ is given as $A_0=(\omega,\dots,\omega)$ and the one describing $\mathcal{G}(U_e,V_e,W_e)$ is given as in equation~\eqref{eq:eul_vec}.
A sequence of switchings taking $A_0$ to $A$ can be found in linear time similarly to the method described in section 4 of~\cite{Bouchet1991}.
The idea is to go over the vectors $A_0$ and $A$ element by element and make these equal one by one.
Let's consider a vertex $v\in V(G)$ and the four vectors $A_0$, $A_0*v$, $A$ and $A*v$.
These four vectors cannot all differ in the element $v$, since there are only three non-zero elements of $\mathbb{F}_4$, i.e. $\{1,\omega,\omega^2\}$.
Repeating this process for all elements of $V(G)$ will give two sequences of switchings, $m_1$ and $m_2$, one for $A_0$ and one for $A$ such that
\begin{equation}
    A_0*m_1=A*m_2.
\end{equation}
Since the switchings are involutions we have that
\begin{equation}
    A_0*(m_1\overline{m}_2)=A
\end{equation}
and therefore that
\begin{equation}
    \tau_{m_1\overline{m}_2}(G)=\mathcal{G}(U_e,V_e,W_e),
\end{equation}
where $m_1\overline{m}_2$ is the sequence $m_1$ followed by the reversal of $m_2$.
Finding $m=m_1\overline{m}_2$ thus takes time $\mathcal{O}(n)$.

\newpage

\bibliographystyle{abbrv}
\bibliography{refs_comb}

\begin{thebibliography}{10}

\bibitem{git}
{Git-repository with implemented code.}
\newblock \url{https://github.com/AckslD/vertex-minors}.

\bibitem{mona}
{MONA}.
\newblock \url{http://www.brics.dk/mona/index.html}.

\bibitem{sage}
{SAGE}.
\newblock \url{http://www.sagemath.org/}.

\bibitem{Aschauer2005}
H.~Aschauer, W.~D{\"{u}}r, and H.~J. Briegel.
\newblock {Multiparticle entanglement purification for two-colorable graph
  states}.
\newblock {\em Physical Review A - Atomic, Molecular, and Optical Physics},
  71(1):1--22, 2005.

\bibitem{Azuma2015}
K.~Azuma, K.~Tamaki, and H.~K. Lo.
\newblock {All-photonic quantum repeaters}.
\newblock {\em Nature Communications}, 6:1--7, 2015.

\bibitem{Bouchet1987}
A.~Bouchet.
\newblock {Isotropic Systems}.
\newblock {\em European Journal of Combinatorics}, 8(3):231--244, 1987.

\bibitem{Bouchet1988a}
A.~Bouchet.
\newblock {Graphic presentations of isotropic systems}.
\newblock {\em Journal of Combinatorial Theory, Series B}, 45(1):58--76, 1988.

\bibitem{Bouchet1991}
A.~Bouchet.
\newblock {An efficient algorithm to recognize locally equivalent graphs}.
\newblock {\em Combinatorica}, 11(4):315--329, dec 1991.

\bibitem{Bouchet1993}
A.~Bouchet.
\newblock {Recognizing locally equivalent graphs}.
\newblock {\em Discrete Mathematics}, 114(1-3):75--86, apr 1993.

\bibitem{Christandl2005}
M.~Christandl and S.~Wehner.
\newblock {Quantum Anonymous Transmissions}.
\newblock In {\em Lecture Notes in Computer Science (including subseries
  Lecture Notes in Artificial Intelligence and Lecture Notes in
  Bioinformatics)}, volume 3788 LNCS, pages 217--235. 2005.

\bibitem{Courcelle1990}
B.~Courcelle.
\newblock {The Monadic Second-Order Theory of Graphs. I. Recognizable Sets of
  Finite graphs}.
\newblock {\em Information and Computation}, 85(1):12--75, 1990.

\bibitem{Courcelle2011}
B.~Courcelle and J.~Engelfriet.
\newblock {\em {Graph Structure and Monadic Second-Order Logic: A Language
  Theoretic Approach}}.
\newblock Cambridge University Press, New York, NY, USA, 1st edition, 2011.

\bibitem{Courcelle2007}
B.~Courcelle and S.~il~Oum.
\newblock {Vertex-minors, monadic second-order logic, and a conjecture by
  Seese}.
\newblock {\em Journal of Combinatorial Theory. Series B}, 97(1):91--126, 2007.

\bibitem{npcomplete}
A.~Dahlberg, J.~Helsen, and S.~Wehner.
\newblock {How to transform graph state using single-qubit operations:
  computational complexity and algorithms}.
\newblock {\em arXiv pre-print: $\mathrm{1805.05306}$}, 2018.

\bibitem{Downey1999}
R.~G. Downey and M.~R. Fellows.
\newblock {Parameterized Complexity}.
\newblock {\em Proc 6th Annu Conf on Comput Learning Theory}, 5(1):51--57,
  1999.

\bibitem{Ganian2010}
R.~Ganian and P.~Hlin{\v{e}}n{\'{y}}.
\newblock {On parse trees and Myhill-Nerode-type tools for handling graphs of
  bounded rank-width}.
\newblock {\em Discrete Applied Mathematics}, 158(7):851--867, 2010.

\bibitem{Gottesman1997}
D.~Gottesman.
\newblock {\em {Stabilizer Codes and Quantum Error Correction}}.
\newblock PhD thesis, may 1997.

\bibitem{Gottesman1999}
D.~Gottesman.
\newblock {The Heisenberg Representation of Quantum Computers}.
\newblock {\em Proceedings of the XXII International Colloquium on Group
  Theoretical Methods in Physics}, 1:32--43, 1999.

\bibitem{Gupta1996}
A.~Gupta and N.~Nishimura.
\newblock {The complexity of subgraph isomorphism for classes of partial
  k-trees}.
\newblock {\em Theoretical Computer Science}, 164(1-2):287--298, 1996.

\bibitem{Hein2006}
M.~Hein, W.~D{\"{u}}r, J.~Eisert, R.~Raussendorf, M.~V. den Nest, and H.~J.
  Briegel.
\newblock {Entanglement in Graph States and its Applications}.
\newblock {\em Quantum Computers, Algorithms and Chaos}, pages 1--99, 2006.

\bibitem{Hein2004}
M.~Hein, J.~Eisert, and H.~J. Briegel.
\newblock {Multiparty entanglement in graph states}.
\newblock {\em Physical Review A - Atomic, Molecular, and Optical Physics},
  69(6):062311--1, 2004.

\bibitem{Oum2006}
S.~il~Oum and P.~Seymour.
\newblock {Approximating clique-width and branch-width}.
\newblock {\em Journal of Combinatorial Theory. Series B}, 96(4):514--528,
  2006.

\bibitem{Jozsa2000}
R.~Jozsa, D.~S. Abrams, J.~P. Dowling, and C.~P. Williams.
\newblock {Quantum Clock Synchronization Based on Shared Prior Entanglement}.
\newblock {\em Physical Review Letters}, 85(9):2010--2013, aug 2000.

\bibitem{Langer2014}
A.~Langer, F.~Reidl, P.~Rossmanith, and S.~Sikdar.
\newblock {Practical algorithms for MSO model-checking on tree-decomposable
  graphs}.
\newblock {\em Computer Science Review}, 13-14:39--74, nov 2014.

\bibitem{Langer2011}
A.~Langer, P.~Rossmanith, and S.~Sikdar.
\newblock {Linear-time algorithms for graphs of bounded rankwidth: A fresh look
  using game theory (extended abstract)}.
\newblock {\em Lecture Notes in Computer Science (including subseries Lecture
  Notes in Artificial Intelligence and Lecture Notes in Bioinformatics)}, 6648
  LNCS:505--516, 2011.

\bibitem{Markham2008}
D.~Markham and B.~C. Sanders.
\newblock {Graph states for quantum secret sharing}.
\newblock {\em Physical Review A - Atomic, Molecular, and Optical Physics},
  78(4), 2008.

\bibitem{Oum2005}
S.~I. Oum.
\newblock {Rank-width and vertex-minors}.
\newblock {\em Journal of Combinatorial Theory. Series B}, 95(1):79--100, 2005.

\bibitem{Raussendorf2001}
R.~Raussendorf and H.~J. Briegel.
\newblock {A One-Way Quantum Computer}.
\newblock {\em Physical Review Letters}, 86(22):5188--5191, may 2001.

\bibitem{Ribeiro2017}
J.~Ribeiro, G.~Murta, and S.~Wehner.
\newblock {Fully device-independent conference key agreement}.
\newblock {\em Physical Review A}, 97(2):022307, feb 2018.

\bibitem{Robertson1986}
N.~Robertson and P.~Seymour.
\newblock {Graph minors: Algorithmic aspects of tree-width}.
\newblock {\em J Algorithms}, 7:309--322, 1986.

\bibitem{VandenNest2004a}
M.~{Van den Nest}, J.~Dehaene, and B.~{De Moor}.
\newblock {Efficient algorithm to recognize the local Clifford equivalence of
  graph states}.
\newblock {\em Physical Review A}, 70(3):034302, sep 2004.

\bibitem{VandenNest2004}
M.~{Van den Nest}, J.~Dehaene, and B.~{De Moor}.
\newblock {Graphical description of the action of local Clifford
  transformations on graph states}.
\newblock {\em Physical Review A}, 69(2):022316, feb 2004.

\bibitem{VandenNest2007}
M.~{Van den Nest}, W.~D{\"{u}}r, G.~Vidal, and H.~J. Briegel.
\newblock {Classical simulation versus universality in measurement-based
  quantum computation}.
\newblock {\em Physical Review A}, 75(1):012337, jan 2007.

\end{thebibliography}

\end{document}